\documentclass[11pt]{article}
\usepackage[a4paper,margin=25mm]{geometry}
\usepackage[T1]{fontenc}
\usepackage{lmodern,amsmath,amssymb,amsthm}
\usepackage{mathrsfs} 
\usepackage[colorlinks=true,linkcolor=blue,citecolor=blue,urlcolor=blue]{hyperref}
\newcommand{\dd}{\,\mathrm d}
\newcommand{\curl}{\nabla\times}
\newcommand{\Div}{\nabla\cdot}
\newcommand{\norm}[1]{\lVert#1\rVert}
\newcommand{\Om}{\Omega}
\newcommand{\Ax}{\mathcal A}
\newcommand{\Kern}{\mathcal K}
\newcommand{\calP}{\mathcal P}
\newcommand{\calM}{\mathcal M}

\newcommand{\calD}{\mathfrak D}

\newcommand{\eps}{\varepsilon}
\theoremstyle{plain}
\newtheorem{theorem}{Theorem}[section]
\newtheorem{proposition}[theorem]{Proposition}
\newtheorem{lemma}[theorem]{Lemma}
\theoremstyle{definition}

\newtheorem{assumption}[theorem]{Assumption}
\theoremstyle{remark}

\hypersetup{pdftitle={A compatibility--realization framework for singular Navier--Stokes flows: admissible families and non-rigid core mechanics},pdfauthor={Weishuo Liu}}
\title{A compatibility--realization framework for\\
singular Navier--Stokes flows:\\
admissible families and non-rigid core mechanics}
\author{Weishuo Liu\\[3pt]
\small School of Mechanics and Engineering Science\\
\small Peking University, Beijing 100871, China\\
\small\href{mailto:liuweishuo@pku.edu.cn}{\texttt{liuweishuo@pku.edu.cn}}}
\date{}
\begin{document}
\maketitle
\begin{abstract}
We introduce a compatibility--realization framework for constructing singular
solutions to the forced three-dimensional incompressible Navier--Stokes
equations. The framework separates exact balance
and matching constraints, expressed in a chosen geometric representation,
from the sufficient conditions of a particular completion scheme. Completed
realizations are organized by observable fibers, distinguishing prescribed
core kinematics from the surrounding force and transport mechanisms.
Conditional on the source-dependent completion input specified in this
paper, we obtain families with zero initial velocity, smooth compactly
supported forcing, bounded kinetic energy and finite-time unbounded
velocity. These include a reversal of leading viscous work at fixed axis
motion and first velocity gradient, a signed transition through a stationary
singular center under one fixed canonical pressure trace, and a connected
stretching family crossing the threshold of zero leading direct viscous
vorticity supply. Near a critical profile, the realizable trace data contain
a neighborhood in an infinite-dimensional space of analytic axial-velocity
and pressure perturbations, while the central first velocity gradient and
the critical vorticity balance remain fixed. We also identify the full
kernel and an explicit right inverse of the joint angular-transport operator
in the chosen realization class. The results establish non-rigidity within
observable fibers and show that axis kinematics alone do not determine the
local pressure--viscosity partition. Direct profile selection and dynamical
stress generation are treated as coupled components of construction, rather
than as definitions of the entire design space.
\end{abstract}

\section{Introduction}
The construction of a singular fluid motion involves more than a growing velocity field. The pressure must agree with the momentum balance, the residual force must have the prescribed regularity, and local concentration must fit the global energy and matching conditions. These requirements are especially restrictive for the three-dimensional incompressible Navier--Stokes equations with positive viscosity. Leray's finite-energy theory, partial regularity and critical-norm criteria describe different aspects of this constraint structure \cite{Leray1934,CKN1982,ESS2003}. The forced breakdown alternatives in the Clay problem statement additionally require a smooth force with specified decay properties \cite{Fefferman2000}.

Several constructive approaches expose freedoms that are not visible in a single velocity ansatz. The differential-inclusion formulation of Euler and subsequent stress-realization methods separate a mean flow from the oscillations used to complete it \cite{DLS2009,Daneri2017}. For Navier--Stokes, convex integration gives finite-energy weak nonuniqueness \cite{BuckmasterVicol2019}, while an unstable similarity background yields nonuniqueness of forced Leray solutions \cite{Albritton2022}. Tao's averaged equation provides a further example in which energy cancellation coexists with finite-time blowup \cite{Tao2016}. These results concern different equations or solution classes, but they demonstrate why an admissible residual and a realizable correction have to be considered together.

Dynamical amplification supplies another part of this picture. Short-wave analysis follows wavevectors and polarizations along a background trajectory \cite{Lifschitz1991,Friedlander1991}; exact shearing-wave solutions retain selected nonlinear cancellations \cite{Craik1986,Singh2017}. Forced Euler and hypodissipative constructions combine amplification with increasingly concentrated scales and control of the force \cite{Cordoba2023,Cordoba2026}. The recent preprint of OpenAI \cite{OpenAI2026} proposes a positive-viscosity construction in which amplified disturbances supply a momentum flux for a collapsing background vortex. Its profile-completion theorem is the external existence input used in this paper.

Here we study the freedom of the completed construction itself. Rather than
prescribing a velocity ansatz and subsequently assigning its residual to the
force, we jointly select the geometry, mean velocity, pressure, cumulative
transport and realization of the required momentum transfer. We call this
joint constrained selection \emph{co-design} and formulate it within a
\emph{compatibility--realization framework}. Exact identities constrain the
actual profile data, while a completion scheme supplies sufficient conditions
under which those data can be realized by a full singular flow. The pressure
normalization, transport operators and stress cone of the source construction
are one concrete instance of this structure, not universal restrictions on
all possible constructions.

Our results concern both the freedom of admissible data and its physical
interpretation. Theorem~\ref{thm:families} gives completed profile families that reverse the
leading viscous work at fixed axis motion and first gradient, pass through
opposite signed material motions under a common canonical pressure trace,
and connect positive, vanishing and negative leading direct viscous
contributions to central vorticity growth. Its analytic deformation result
also gives a local lifting statement: a full neighborhood of prescribed
analytic velocity and pressure traces is attained within a fiber of fixed
central kinematics. Appendix~\ref{app:angular} identifies an explicit right inverse and the
complete function-valued kernel of the angular-transport compatibility
operator. These statements distinguish genuine freedom in the singular
profile from the merely formal freedom to introduce parameters into an
ansatz.

An observation map organizes these constructions by the core information
that they share. A force or transport diagnostic is identifiable from those
observations precisely when it is constant on the corresponding fibers.
The paired constructions show that the pressure--viscosity partition is not
identifiable from axis velocity and its first gradient. Different comparisons
use different observation maps; in particular, varying the central stretching
rate does not preserve the complete central first gradient. All full-flow
existence statements are conditional on Assumption~\ref{ass:completion},
with the source-dependent interface and matching requirements specified in
Appendix~\ref{app:completion}. The conclusions require a realization for each admissible
profile, not a continuous selection of final oscillatory fields, and the
forcing is allowed to vary between realizations.

\section{The compatibility--realization framework}
\label{sec:framework}
Consider
\begin{equation}
 \partial_tu+(u\cdot\nabla)u-\nu\Delta u+\nabla p=f,
 \qquad\Div u=0,\qquad u(\cdot,0)=0,
 \label{eq:NS}
\end{equation}
on $\mathbb R^3\times[0,1)$, with fixed $\nu>0$. The target class has
\begin{equation}
 f\in C_c^\infty(\mathbb R^3\times(0,\infty)),\qquad
 \sup_{t<1}\|u(t)\|_2<\infty,\qquad
 \limsup_{t\uparrow1}\|u(t)\|_\infty=\infty.
 \label{eq:target}
\end{equation}
Smoothness of $f$ includes extension through $t=1$. The explicit profiles
below use viscosity-one coordinates; the rescaling in Appendix~\ref{app:axis} gives every
positive viscosity.

\subsection{Compatibility, realization and completion}
\label{subsec:compatibility}
Fix a geometric representation $\mathfrak g$, and let $\mathscr X_{\mathfrak g}$
be its space of profile data. A design $\mathbf d$ records actual mean
velocity and pressure fields, cumulative transports, and any residual
stress used in the construction. Its exact compatibility locus is
\begin{equation}
 \mathcal C_{\mathfrak g}
 =\{\mathbf d\in\mathscr X_{\mathfrak g}:
       \mathcal E_{\mathfrak g}(\mathbf d)=0\}.
 \label{eq:compatibility}
\end{equation}
The equations $\mathcal E_{\mathfrak g}=0$ include incompressibility and the
pressure, transport and matching identities appropriate to that
representation. They are independent of a particular recipe for realizing
the stress, but not independent of the chosen geometric representation.
In particular, the canonical pressure formula used below is an identity of
our anisotropic profile class, not a pressure formula asserted for arbitrary
Navier--Stokes flows.

The coupling to realization can be seen directly in the momentum residual
\begin{equation}
 \mathcal R_\nu(v,q)
 =\partial_tv+\nabla\cdot(v\otimes v)-\nu\Delta v+\nabla q.
 \label{eq:residual}
\end{equation}
If a mean profile satisfies
$\mathcal R_\nu(\bar u,\bar p)=\nabla\cdot R+f_{\rm rem}$, a correction
$(w,\pi)$ produces
\begin{align}
 \mathcal R_\nu(\bar u+w,\bar p+\pi)
 &=f_{\rm rem}+\nabla\cdot R
    +\mathcal K_\nu[\bar u;w,\pi],
 \label{eq:coupled-residual}\\
 \mathcal K_\nu[\bar u;w,\pi]
 &=\partial_tw-\nu\Delta w+\nabla\pi
   +\nabla\cdot(\bar u\otimes w+w\otimes\bar u+w\otimes w).
 \nonumber
\end{align}
Thus a stress decomposition is not itself a realization. Actual corrections
must cancel the nonsmooth part of the residual while their interactions,
pressure and localization errors satisfy the same smooth-force requirement.
The final velocity must also retain the initial condition, bounded energy
and blowup in \eqref{eq:target}.

Let $\rho$ label a completion scheme and let $\xi$ denote its realization
variables. Write its joint feasible set as
\begin{equation}
 \begin{split}
 \mathfrak Z_\rho=\bigl\{(\mathbf d,\xi):\;&
 \mathbf d\in\mathcal C_{\mathfrak g(\rho)},\quad
 \mathfrak F_\rho(\mathbf d,\xi)=0,\\
 &\mathcal I_\rho(\mathbf d,\xi)>0,\quad
 B_{\rho,k}(\mathbf d,\xi)<\infty\quad(k\geq0)\bigr\},
 \qquad
 \calD_\rho=\operatorname{pr}_{\mathbf d}\mathfrak Z_\rho.
 \end{split}
 \label{eq:designset}
\end{equation}
Here $\mathfrak F_\rho$ records stress production and the remaining
higher-order closure and matching equations. The inequalities are
scheme-dependent realizability and nondegeneracy conditions; at a flat
stress edge they concern the normalized direction. The bounds record the
regularity, support and smooth-extension estimates required by that scheme.
These symbols organize the obligations of a construction; they do not
supply a new completion theorem.

For a scheme with a valid completion statement, define the correspondence
\begin{equation}
 \mathbf{Comp}_\rho(\mathbf d)
 =\{(u,p,f): (u,p,f)\text{ is a full realization of }\mathbf d
                  \text{ through }\rho\}.
 \label{eq:completion-correspondence}
\end{equation}
Its completion statement asserts
$\mathbf{Comp}_\rho(\mathbf d)\neq\varnothing$ on the specified admissible
set. Neither uniqueness nor continuous selection is part of this definition.
The collection $\{\calD_\rho,\mathbf{Comp}_\rho\}_{\rho\in\mathfrak R}$
allows different realizations to be compared without identifying one
scheme's stress cone with the whole design problem. Only the reference
scheme $\rho_0$ is used for the existence results here, conditionally on
Assumption~\ref{ass:completion}.

\subsection{Observable fibers and admissible lifting}
\label{subsec:fibers}
Let $\mathscr S_{\rm real}$ denote the full solutions obtained through the
specified completion schemes. For an observation map
$\mathcal O:\mathscr S_{\rm real}\to\mathcal Y$, define
\begin{equation}
 \mathscr S_o=\{s\in\mathscr S_{\rm real}:\mathcal O(s)=o\}.
 \label{eq:solution-fiber}
\end{equation}
Observations may retain late axis trajectories, velocity gradients or
selected growth coefficients. A diagnostic
$\mathcal Q:\mathscr S_{\rm real}\to\mathcal Z$ is identifiable from
$\mathcal O$ when it is constant on every nonempty observation fiber.
Equivalently,
\begin{equation}
 \mathcal Q=\widehat{\mathcal Q}\circ\mathcal O
 \quad\text{for a map }\widehat{\mathcal Q}
 \text{ defined on }\mathcal O(\mathscr S_{\rm real}).
 \label{eq:identifiability}
\end{equation}
This is a set-theoretic notion; no continuity of
$\widehat{\mathcal Q}$ is implied. A pair with the same observations and
different diagnostic values disproves identifiability for that observation
map.

When the completion preserves the chosen observables, they descend to the
profile data and define
\begin{equation}
 \calD_{\rho,o}
 =\{\mathbf d\in\calD_\rho:\mathcal O(\mathbf d)=o\}.
 \label{eq:fiber}
\end{equation}
Given a trace map $\mathcal T$ and a reference member $\mathbf d_0$, an
\emph{admissible lifting} of prescribed trace data $a$ is a profile
$\mathbf d\in\calD_{\rho,o}$ with $\mathcal T(\mathbf d)=a$.
A local lifting result therefore asserts that
$\mathcal T(\calD_{\rho,o})$ contains a neighborhood of
$\mathcal T(\mathbf d_0)$ in the stated trace space. This is stronger than
listing formal parameters: the exact constraints, realization conditions
and completion estimates have to be satisfied for every admitted trace.

For a fixed scheme and a differentiable function-space formulation, set
$x=(\mathbf d,\xi)$ and
$\mathcal H_\rho(x)=(\mathcal E_{\mathfrak g}(\mathbf d),
\mathfrak F_\rho(\mathbf d,\xi))$. The equality-preserving linearized
directions at $x_0$ are
\begin{equation}
 \mathcal V^{\rm lin}_{\rho,x_0,o}
 =\ker D\mathcal H_\rho(x_0)\cap\ker D\mathcal O_\rho(x_0),
 \label{eq:linearized-directions}
\end{equation}
where $\mathcal O_\rho$ is the observation written in the same variables.
This notation is used only when the indicated derivatives exist in the
chosen spaces. It does not identify these directions with an actual tangent
space or assert that they integrate to completed families. Nonlinear
lifting, the realizability inequalities and all required bounds remain
separate obligations.

Theorem~\ref{thm:families}(\ref{item:family-partition}) gives two completed realizations with identical full axis
kinematics and first gradient but different pressure--viscosity partitions.
Theorem~\ref{thm:families}(\ref{item:family-analytic}) gives a nonlinear lifting result in an infinite-dimensional
space of analytic traces at fixed central first gradient. The
stretching-rate family uses a weaker observation, namely the common central
vorticity trace at fixed $h$ and center swirl coefficient; its complete
first gradients vary with the stretching rate.

\subsection{An explicit anisotropic realization class}
\label{subsec:explicit}
Fix a small $h>0$ and set
\begin{equation}
 \begin{gathered}
 A=\tfrac12+h,\qquad D=\tfrac12-h,\qquad
 d=1-\eta^2,\qquad L=1-2h\eta^2,\\
 1-t=qd,\qquad z=q^D\eta,\qquad X=\frac{r^2}{2q}.
 \end{gathered}
 \label{eq:coordinates}
\end{equation}
The scalar $d$ in these coordinates is distinct from the design
$\mathbf d$. The leading azimuthal and axial velocities are $q^{-A}E(X,\eta)$
and $q^{-A}U(X,\eta)$. Set $H=\sqrt{2X}E$. Their actual radial integrals and
canonical pressure are
\begin{align}
 M&=\int_0^X U\,dx,&
 I&=\int_0^X H\,dx,&
 J&=\int_0^X UH\,dx,&
 S&=\int_0^X(U^2-E^2/2)\,dx,
 \label{eq:moments}\\
 \Pi(X,\eta)&=-\int_X^\infty\frac{E(x,\eta)^2}{2x}\,dx.
 \label{eq:pressure}
\end{align}
Thus a change in canonical pressure requires a change in the actual swirl
profile. Introduce $y=\log X$ and
\[
 \mu=M/X,\qquad \iota=I/(X\sqrt{2X}),\qquad
 \jmath=J/(X\sqrt{2X}),\qquad \sigma=S/X.
\]
With the prescribed shears $a=1-2E_y/E$ and $b=2U_y/E$, the exact forward
system is
\begin{equation}
 \begin{aligned}
 E_y&=(1-a)E/2,& U_y&=bE/2,\\
 \mu_y&=U-\mu,& \iota_y&=E-3\iota/2,\\
 \jmath_y&=UE-3\jmath/2,& \sigma_y&=U^2-E^2/2-\sigma,\\
 \Pi_y&=E^2/2.
 \end{aligned}
 \label{eq:forward}
\end{equation}
These equations contain no parameter derivatives until a particular
feedback law is imposed. The actual residual sources do contain those
derivatives:
\begin{align}
 W&=1-2D\eta\mu-d\mu_\eta,\nonumber\\
 Q&=-W+\frac{(1-h)\iota-D\eta\iota_\eta-d\jmath_\eta
                   +2(h-D)\eta\jmath}{E},
 \label{eq:Q}\\
 N&=-WU+D(\mu-\eta\mu_\eta)+4h\eta\sigma-d\sigma_\eta
                         +4A\eta\Pi-d\Pi_\eta.
 \label{eq:N}
\end{align}
Incompressibility fixes the radial velocity through
\begin{equation}
 r u_r^{(0)}=\frac{X}{L}
       (2\eta U-2D\eta\mu-d\mu_\eta).
 \label{eq:radial}
\end{equation}
The pressure, radial flow and residual therefore cannot be reset between
design stages.

For the reference pulse realization of \cite{OpenAI2026}, set
\begin{equation}
 \begin{gathered}
 p_1=XQ/L,\qquad p_2=XN/(LE),\qquad
 \vartheta=-b/a,\qquad v=a(1+\vartheta^2),\\
 P_c=p_1+\vartheta p_2,\qquad J_c=p_2-\vartheta p_1.
 \end{gathered}
 \label{eq:cone-coordinates}
\end{equation}
On the active stress annulus its sufficient cone is
\begin{equation}
 a>0,\qquad v>2,\qquad P_c>v,\qquad
 (v-2)J_c^2<2(P_c-v)^2.
 \label{eq:cone}
\end{equation}
The actual tangential stress is
\begin{equation}
 T_0=\frac{E}{\sqrt{2X}}\bigl((p_1,p_2)-(a,-b)\bigr).
 \label{eq:stress}
\end{equation}
The inner core satisfies $T_0=0$ and is regular at the axis. The cone
coordinates are used on the active annulus; a sign requirement there need
not hold at every radius of the stress-free core.

The specified terminal heat profile requires $M(\infty)=S(\infty)=0$ and
a joint angular-transport condition. With
$H_{\rm pow}=\sqrt2c_\infty X^{-h}$, define
\begin{equation}
 I_0=\int_0^\infty(H-H_{\rm pow})\,dX,\qquad
 J_0=\int_0^\infty UH\,dX.
 \label{eq:angular-moments}
\end{equation}
The compatibility equation is
\begin{equation}
 \mathcal A_0(I_0,J_0)
 :=D\eta I_0'-(1-h)I_0+dJ_0'+(1-4h)\eta J_0=0.
 \label{eq:angular}
\end{equation}
It is a constraint on the pair, not two independent zero-moment
requirements. It admits a full function of freedom:
\begin{equation}
 (I_0,J_0)=\left(\frac{d\phi'+2k_0\eta\phi}{L},
                  \frac{k_0\phi-D\eta\phi'}{L}\right),
 \qquad k_0=\tfrac32-2h.
 \label{eq:angularpotential}
\end{equation}
More generally, Appendix~\ref{app:angular} gives
\begin{equation}
 \begin{aligned}
 \mathcal R_0\chi
 &=\frac{\chi}{k_0L}(-d,D\eta),\\
 \mathcal K_0\varphi
 &=\frac1L(d\varphi'+2k_0\eta\varphi,
                    k_0\varphi-D\eta\varphi'),\\
 \mathcal A_0\mathcal R_0\chi&=\chi,\qquad
 \ker\mathcal A_0=\mathcal K_0 C^\infty([-1,1]).
 \end{aligned}
 \label{eq:angular-splitting}
\end{equation}
Hence all smooth solutions of $\mathcal A_0(I,J)=\chi$ have the form
$(I,J)=\mathcal R_0\chi+\mathcal K_0\varphi$. This explicit decomposition
separates a compulsory compatibility correction from function-valued
transport freedom. It is an algebraic statement: an arbitrary kernel
choice is not automatically compatible with the pressure, stress cone and
completion estimates. Appendix~\ref{app:angular} verifies the additional small nonzero
angular-budget family used here.

At the finite matching stage, Appendix~\ref{app:completion} restores simultaneous moment
changes using a full-rank moment differential and its right inverse,
\begin{equation}
 B(\eta)c+Q_\eta(c,c)=\delta m(\eta),\qquad
 \mathcal R_B(\eta)=B(\eta)^T[B(\eta)B(\eta)^T]^{-1}.
 \label{eq:moment-lifting}
\end{equation}
The uniform rank bound, the parameter norms, analytic core reconstruction
and normalized edge margins are part of that argument. These concrete
lifting operations, rather than the abstract notation alone, support the
profile deformations below.

\begin{assumption}[Completion input]
\label{ass:completion}
We take as an external input the validity, in their stated form, of the
analytic-core estimates, stress-realization and whole-space completion results of \cite[Theorem~4.6, Sections~5--10 and Appendices~A--C]{OpenAI2026} through the interface specified in Appendix~\ref{app:completion}. A leading
profile with the required regular axis, flat stress edges, realizable
annular stress, exact exterior matching and ordered correction intervals,
including the stated moment solvability and parameter bounds, admits a full
solution satisfying \eqref{eq:NS}--\eqref{eq:target}. The completion
preserves the prescribed velocity and pressure traces and the first
velocity gradient on the axis in a late local neighborhood. The force and
its derivatives are flat there at the singular time.
\end{assumption}
All full-flow existence statements in this paper are under this assumption.
The displayed profile, transport and axis identities are direct
calculations. The added reconstruction arguments are given in
Sections~\ref{subsec:app-core-contraction}--\ref{subsec:app-partition-reconstruction},
Lemma~\ref{lem:prepared-rank}, and Proposition~\ref{prop:angular-completion}.
They verify the changed data before the source completion is invoked.

\section{Realized families and their independent freedoms}\label{sec:families}
Write the prescribed axial trace as $G(\eta)=U(0,\eta)$ and the swirl normalization as $F_*(\eta)=\lim_{X\downarrow0}E/\sqrt{2X}$. The following results concern actual globally matched profiles. Their proofs reconstruct the equalities as well as preserving the strict inequalities.

\begin{theorem}[Families of completed singular profiles]\label{thm:families}
Under Assumption~\ref{ass:completion}, the following designs admit full Navier--Stokes realizations satisfying \eqref{eq:target}.
\begin{enumerate}
\renewcommand{\labelenumi}{(\roman{enumi})}
\renewcommand{\theenumi}{\roman{enumi}}
\item\label{item:family-partition} For $G=4\eta+j_0$ with a sufficiently small $j_0>0$, a reference design and a pressure-modified design can have the same $G,F_*$ and the same complete first velocity gradient throughout the prescribed late axis region. On their common distinguished material particle, the reference viscosity performs negative leading work, whereas the modified pressure can satisfy
\begin{equation}
 -p_z=-\beta a_H,\qquad \nu\Delta u_z+f_z=(1+\beta)a_H,
 \qquad\beta\ge0,
 \label{eq:partition}
\end{equation}
for any prescribed finite $\beta$, where $a_H>0$ is the common material acceleration. Common axis data can be selected for any fixed compact range of such $\beta$, together with the reference member.
\item\label{item:family-offset} There are a fixed nonzero tilt $\gamma>0$ and $j_*>0$ such that one canonical axis pressure trace $\Pi_\gamma=\Pi_{\rm ref}+\gamma\eta$ supports every $G_j=4\eta+j$, $|j|\le j_*$. The distinguished axis particle has upward acceleration for $j>0$, is stationary for $j=0$, and has downward acceleration for $j<0$.
\item\label{item:family-stretching} There are $0<h_-<h_+<1/100$ and $\Delta_->0$ for which every
\begin{equation}
 h\in[h_-,h_+],\qquad 1+h-\Delta_-\le m\le4
 \label{eq:mfamily}
\end{equation}
has a realization with $G=m\eta$. At fixed $h$ and nonzero pressure tilt, these profiles can share the canonical axis pressure and the center swirl coefficient $C^{-1}$. The interval can be chosen to include $m=1$.
\item\label{item:family-analytic} Around a fixed completed critical member $G_0=(1+h)\eta$, there are a bounded common complex neighborhood $\Om$ and $\eps_F>0$ such that
\begin{equation}
 G=G_0+g,\quad\Pi_{\rm ax}=\Pi_{{\rm ax},0}+\psi,\quad F_*=F_{*,0},
 \qquad \norm g_\Om+\norm\psi_\Om<\eps_F,
 \quad g(0)=g'(0)=0
 \label{eq:functionfamily}
\end{equation}
is realizable for bounded holomorphic $g,\psi$ that are real on the real interval.
\end{enumerate}
The parameter families refer to admissible profiles, each with a full realization. Their forces may differ; a continuous selection of the final oscillatory fields is not needed for these assertions.
\end{theorem}
\begin{proof}
Proposition~\ref{prop:critical-connection} proves
part~\ref{item:family-stretching} by retaining the radial strain increment
through the critical connection. Section~\ref{subsec:app-local-lifting}
proves part~\ref{item:family-analytic} by an actual parameter-dependent
core solve, moment restoration and mean-preserving stress modulation.
Section~\ref{subsec:app-partition-reconstruction} proves
parts~\ref{item:family-partition} and~\ref{item:family-offset}, without
interpolating through a degenerate exit. The rank and parameter inverses
used in all these steps are supplied by Lemma~\ref{lem:prepared-rank}
and \eqref{eq:app-analytic-moment-lift}. Appendix~\ref{app:completion}
then identifies the external full-flow completion applied to each member.
\end{proof}

\paragraph{Infinite-dimensional lifting within a central-kinematics fiber.}
Theorem~\ref{thm:families}(\ref{item:family-analytic}) can be expressed as a local surjectivity statement for
prescribed traces. Fix its critical reference profile $\mathbf d_0$, the
parameters $h$ and $F_*$, and let $\mathcal O_c$ record the late central
velocity and first velocity gradient. Write $o_c=\mathcal O_c(\mathbf d_0)$
and define the real Banach space
\begin{equation}
 \begin{split}
 \mathscr H_\Omega^0=\{(g,\psi)\in
 H^\infty(\Omega)\times H^\infty(\Omega):\;&
 g,\psi\text{ are real on }[-1,1],\\
 &g(0)=g'(0)=0\},
 \qquad
 \|(g,\psi)\|=\|g\|_\Omega+\|\psi\|_\Omega.
 \end{split}
 \label{eq:trace-space}
\end{equation}
For $\mathcal T_0(\mathbf d)=(G-G_0,
\Pi_{\rm ax}-\Pi_{{\rm ax},0})$, that theorem gives
\begin{equation}
 B_{\eps_F}(0;\mathscr H_\Omega^0)
 \subseteq
 \mathcal T_0\bigl(\calD_{\rho_0,o_c}\bigr),
 \label{eq:trace-surjectivity}
\end{equation}
where the admissible fiber is restricted to the fixed $h$ and $F_*$.
Indeed, each pair in the ball is realized by the profile asserted in
Theorem~\ref{thm:families}(\ref{item:family-analytic}), and its prescribed central first gradient is unchanged.
Thus the freedom is infinite-dimensional at the level of realized trace
data. The inclusion does not assert that the complete solution set is a
Banach manifold or that a continuous right inverse into the final
oscillatory fields has been constructed. It also does not assert that the
entire axis first-gradient trace remains fixed under arbitrary $g$.

The first two families alter force allocation and material motion with $m=4$. The third changes the dimensionless strain itself, including the critical and subcritical cases. The fourth varies higher spatial organization after the central first gradient has been fixed. In particular, the analytic family is not confined to a predetermined polynomial ansatz.

For an explicit closed slice of \eqref{eq:functionfamily}, let $R_\Om=\max(1,\sup_\Om|\eta|)$ and define
\begin{equation}
 g=\kappa\eta^3,\qquad\psi=\rho\eta^2/2,\qquad
 |\kappa|\le\frac{\eps_F}{4R_\Om^3},\qquad
 |\rho|\le\min\left(\frac{\eps_F}{2R_\Om^2},\frac{J_*}{2}\right),
 \label{eq:rectangle}
\end{equation}
where
\[
 J_*:=\Pi_{{\rm ax},0}''(0)-4A\Pi_{{\rm ax},0}(0)
                 +(1+h)(2+h)>0
\]
is the critical reference strain-supply coefficient from \eqref{eq:strainbalance}. The closed rectangle remains strictly inside the reconstruction neighborhood. Its radius is specified by the finite construction constants, rather than by an unverified numerical tolerance.

There is also freedom in the integral transport. A small nonzero angular-budget parameter can be realized with
\begin{equation}
 I_0=c\left((1-h)^{-1}-4\eta^2\right),\qquad J_0=c\eta.
 \label{eq:nonzeroangular}
\end{equation}
The first correction flux explicitly supplies its axial viscous transport, while the higher corrections can have zero angular moments. Proposition~\ref{prop:angular-completion} realizes this pair by actual leading-profile moment edits and its viscous return by actual higher-order velocity corrections. It then checks the stress supports, flat residual summation and the downstream completion interfaces. This supplies a completed transport deformation, not just an algebraic assignment of a mixed flux.

\section{Core physics: motion, force and vorticity supply}\label{sec:physics}
\subsection{Axis motion does not determine the pressure--viscosity partition}
For general $G$ define
\begin{equation}
 H_*=D\eta+dG,\qquad
 \calM=A(1-2\eta G)G+H_*G',\qquad
 \calP=4A\eta\Pi_{\rm ax}-d\Pi_{\rm ax}',\qquad Z_* =\calP-\calM.
\end{equation}
The exact late axis traces give
\begin{equation}
 (D_tu)_z=\frac{q^{-A-1}}L\calM,\qquad
 -p_z=\frac{q^{-A-1}}L\calP,\qquad
 \Delta u_z+f_z=-\frac{q^{-A-1}}LZ_*.
 \label{eq:axisforces}
\end{equation}
For a root $\eta_H$ of $H_*$, the fixed-similarity curve
\begin{equation}
 z_H(t)=\eta_H\left(\frac{1-t}{d_H}\right)^D,
 \qquad G_H=-D\eta_H/d_H
 \label{eq:particle}
\end{equation}
is an actual material trajectory. Its acceleration is $a_H=A G_Hq_H^{-A-1}/d_H$.

With $G=4\eta+j_0$, $j_0>0$, this particle approaches the origin from below. The source reference has $Z_*(\eta_H)>0$, so pressure supplies more than the net acceleration and viscosity brakes the particle. The modification in Theorem~\ref{thm:families}(\ref{item:family-partition}) makes the pressure vanish or oppose the same acceleration. It leaves the full first axis gradient unchanged and changes the higher radial derivatives that enter the Laplacian. Thus identical local velocity, strain and rotation can accompany opposite signs of viscous work.

This observation is compatible with global dissipation. For a smooth incompressible field, with $S_u=(\nabla u+\nabla u^T)/2$,
\begin{equation}
 \nu u\cdot\Delta u=\Div(2\nu S_u u)-2\nu|S_u|^2.
 \label{eq:localenergy}
\end{equation}
The paired axis gradients have the same irreversible dissipation density $2\nu|S_u|^2$. Their viscous transport term can nevertheless have different signs. Spatial integration recovers the usual negative viscous energy contribution \cite{ConstantinFoias1988}. Positive local viscous work therefore represents redistribution into the sampled region, not a violation of the energy balance.

\subsection{A stationary singular center and signed material motion}
For the fixed-tilt family, $H_j/d=D\eta/d+4\eta+j$ is strictly increasing. Its unique zero obeys
\[
 \eta_j=-\frac{j}{D+4}+O(j^3),\qquad
 G_j(\eta_j)=\frac{D}{D+4}j+O(j^3).
\]
For sufficiently small $|j|$, $Z_j(\eta_j)<-\gamma/2$. The viscous acceleration points upward throughout the family. When $j>0$ it drives the upward motion against pressure; when $j<0$ it brakes the downward motion driven by pressure. The leading pressure and viscous powers are of opposite order $j(1-t)^{-2A-1}$, leaving an order $j^2(1-t)^{-2A-1}$ positive net gain.

At $j=0$, the origin is an exact stationary material particle although its gradient diverges. The axial force balance is
\begin{equation}
 u(0,t)=D_tu(0,t)=0,
 \qquad -p_z(0,t)=-\gamma(1-t)^{-A-1},\qquad
 \nu\Delta u_z(0,t)+f_z(0,t)=\gamma(1-t)^{-A-1}.
 \label{eq:stationary}
\end{equation}
These forces do no work on a particle with zero velocity. The full velocity supremum still diverges through the positive swirl profile at shrinking nonzero radii. Each fixed nonzero axis particle stays away from the origin under the outward axis motion and remains individually regular. This distinguishes a singular spatial supremum from divergent speed on every axis trajectory.

\subsection{Stretching and direct vorticity diffusion can be separated}
Put $\tau=1-t$. For the unbiased family $G=m\eta$, the common normalized swirl has $F_*(0)=C^{-1}$ and the completed center gradient is
\begin{equation}
 \nabla u(0,t)=
 \begin{pmatrix}
 -m/(2\tau)&-C^{-1}\tau^{-1-h}&0\\
 C^{-1}\tau^{-1-h}&-m/(2\tau)&0\\
 0&0&m/\tau
 \end{pmatrix}.
 \label{eq:centergradient}
\end{equation}
Consequently $\omega_z(0,t)=2C^{-1}\tau^{-1-h}$. The exact vorticity equation yields
\begin{equation}
 \boxed{\nu\Delta\omega_z+(\curl f)_z
       =\frac{1+h-m}{\tau}\omega_z.}
 \label{eq:vorticitybalance}
\end{equation}
The force curl is flat, so the coefficient determines the leading direct viscous contribution. It is negative for $m>1+h$, zero for $m=1+h$, and positive for the realizable subcritical interval. For $m=1$, stretching supplies coefficient $1$ of the growth rate $1+h$, and direct diffusion supplies the remaining positive coefficient $h$.

Viscous supply of strain is a different quantity. With $s=\partial_z u_z$ and
\[
 K_p=\Pi_{\rm ref}''(0)-4A\Pi_{\rm ref}(0)>0,
\]
the differentiated axial balance gives
\begin{equation}
 \nu\Delta s+\partial_zf_z
   =[K_p+m(m+1)]\tau^{-2}>0.
 \label{eq:strainbalance}
\end{equation}
This positive strain-supply coefficient persists across the sign change in \eqref{eq:vorticitybalance}. At the critical member the leading direct vorticity diffusion vanishes, while viscosity still participates in maintaining the strain that stretches vorticity.

The subcritical profiles have a corresponding local shape. Writing $\delta=m-1-h$, their analytic core satisfies $\varphi_X(0,0)/\varphi(0,0)=-\delta/4$. For $\delta<0$, angular velocity initially increases with radius squared and the leading radial vorticity curvature is positive. A short interval of negative azimuthal shear occurs inside the stress-free core. The construction places its active edge after positive shear has been recovered; that inner sign change is consistent with the annular realization conditions.

In the analytic slice \eqref{eq:rectangle}, write
$\delta_{\rm fam}F:=F_{\kappa,\rho}-F_{0,0}$ for a finite family difference,
not the spatial Laplacian $\Delta$. The central first gradient and critical
balance remain fixed, but
\begin{equation}
 \delta_{\rm fam}(\partial_z^3u_z)(0,t)=6\kappa\tau^{-2+2h},\qquad
 \nu\Delta s+\partial_zf_z=(J_*+\rho)\tau^{-2}
 \label{eq:higherfreedom}
\end{equation}
in normalized units. The first parameter changes higher axis geometry; the second changes pressure curvature and the amount of strain supplied through viscosity. Neither changes the center pressure value or its prescribed linear tilt.

\subsection{What a proposed mechanism restriction must resolve}
These families test three possible pointwise restrictions: that a singular material particle must be accelerated by pressure, that viscosity must locally brake it, and that central vorticity growth must have a fixed sign of direct viscous supply. The completed examples realize the alternatives just described. A mechanism criterion based only on axis velocity and its first gradient therefore cannot uniquely assign these force channels.

This conclusion concerns pointwise core diagnostics. Geometric regularity criteria such as the vorticity-direction condition of Constantin and Fefferman \cite{ConstantinFefferman1993} control relations across high-vorticity regions; they are not contradicted by agreement of one axis gradient. The construction instead identifies information that a more complete diagnostic may need: higher radial organization, canonical pressure, cumulative transport and the spatial distribution of the realized stress.

\section{Implementation strategies within the design space}
\label{sec:directions}
The compatibility--realization framework distinguishes profile selection from stress
generation without treating them as mutually exclusive classes of
solutions. A profile-led strategy first proposes geometry, velocity and
pressure data, and then seeks compatible transport and an actual realization
of the required stress. A dynamics-led strategy first selects an amplifying
background or disturbance mechanism, and then seeks profile data whose
momentum demand can be supplied by the resulting fluctuations. Short-wave,
shearing-wave, periodic or transient amplification can motivate such choices,
but an amplification calculation is not by itself a singular construction;
relevant precedents are \cite{Lifschitz1991,Friedlander1991,Craik1986,Albritton2022}.

Within a fixed formulation, these are different entry points into the
coupled equations for $(\mathbf d,\xi)$. They are not asserted to be
interchangeable algorithms, and neither a compatible mean profile nor a
candidate stress generator closes the problem alone. A directly selected
profile can still require dynamically amplified corrections; a selected
amplification mechanism must still satisfy the actual pressure, transport,
matching and smooth-force constraints. Each stage must pass its actual
state---including cumulative moments, parameter derivatives and remaining
correction resources---to the next stage.

The explicit pressure formula, angular operator and moment correction in
Section~\ref{subsec:explicit} belong to the anisotropic class used here.
Changing geometry or the completion mechanism requires identifying and
verifying the corresponding compatibility and realization conditions anew.
The framework records these obligations without claiming that the same
operator, cone or pulse construction applies to every geometry. The
completed results established here remain the families of Theorem~\ref{thm:families} and
the compatible angular-budget extension; alternative global realizations
are not existence results of the present paper.

\section{Conclusion}
\label{sec:conclusion}
The compatibility--realization framework treats singular-flow construction as a coupled problem:
the prescribed core, its pressure and transport budgets, and the mechanism
that supplies the required momentum transfer must be selected together.
The source-dependent completion used here realizes several non-rigid
families, including opposite signs of leading viscous work at fixed axis
kinematics, signed material motion under a common canonical pressure trace,
and a connected transition in the leading direct viscous supply of central
vorticity.

The analytic deformation theorem gives a local lifting of an
infinite-dimensional neighborhood of prescribed traces within a fixed
central-kinematics fiber. The explicit angular right inverse and kernel
separate compulsory transport compatibility from additional design freedom.
These are concrete statements about admissible profiles with full
realizations under the stated completion input, not a classification or a
manifold theorem for all singular solutions. Observable fibers provide a
way to express what a proposed core diagnostic can and cannot determine,
while the compatibility--realization formulation separates the obligations
shared by a construction problem from the sufficient conditions of one
implementation.

\paragraph{Verification and availability.}
The appendices give the reconstruction arguments and the precise completion input. The ancillary script \texttt{verify\_identities.py} checks the displayed algebraic identities symbolically. The existence arguments are the analytic core and connection estimates, the actual moment inverses and the angular-background extension proved in the appendices; they are not consequences of the symbolic tests. The external completion results of \cite{OpenAI2026} remain the explicit input in Assumption~\ref{ass:completion}.

\paragraph{Acknowledgment of AI assistance.}
OpenAI ChatGPT and Codex assisted exploratory derivations, literature discovery, scientific interpretation, drafting, code generation and consistency checks. The author directed the research questions and scope. Verification was assisted by AI. The human author is responsible for the final manuscript.

\appendix
\section{Axis identities and completed profile families}\label{app:axis}

This appendix records the axis calculations underlying the comparisons in Theorem~\ref{thm:families}. We distinguish exact trace identities from the source-dependent realization of the profiles to which they apply. We work first in the viscosity-one normalization of \eqref{eq:coordinates}; the physical viscosity is restored at the end. Sections~\ref{subsec:app-core-contraction}--\ref{subsec:app-partition-reconstruction} supply the core and global reconstruction proofs. Full-flow realization then uses Assumption~\ref{ass:completion} through Appendix~\ref{app:completion}. The preserved observations are those of Section~\ref{subsec:fibers}.

\subsection{Axis force decomposition}\label{subsec:app-axis-force}
Let $G(\eta)$ be the axial trace and let $\Pi_{\rm ax}(\eta)$ be the canonical pressure trace in the late, uncut axis neighborhood. Define
\begin{equation}
 H_*=D\eta+dG,\qquad
 \calM=A(1-2\eta G)G+H_*G',\qquad
 \calP=4A\eta\Pi_{\rm ax}-d\Pi_{\rm ax}',
 \qquad Z_*:=\calP-\calM .
 \label{eq:app-axis-def}
\end{equation}
The coordinate differentiation in the source completion gives the exact axis balances
\begin{equation}
 (D_tu)_z=\frac{q^{-A-1}}{L}\calM,\qquad
 -\partial_zp=\frac{q^{-A-1}}{L}\calP,\qquad
 \nu\Delta u_z+f_z=-\frac{q^{-A-1}}{L}Z_* .
 \label{eq:app-axis-balance}
\end{equation}
The last identity is the full momentum equation evaluated on the completed axis. It includes the pressure derivative and all positive-order radial corrections; it is stronger than retaining only the leading radial Laplacian. Statements about the viscous term alone use the flatness of $f$ and its derivatives from Assumption~\ref{ass:completion}. The displayed sums containing $f$ are exact, whereas the sign conclusions for viscosity alone are leading asymptotic statements as $t\uparrow1$; flatness does not set $f$ identically to zero for $t<1$.

Suppose that $\eta_H$ is a zero of $H_*$. Then the axis characteristic through this trace is
\begin{equation}
 z_H(t)=\eta_H\left(\frac{\tau}{d_H}\right)^D,\qquad
 q_H=\frac{\tau}{d_H},\qquad \tau=1-t,\qquad d_H=1-\eta_H^2,
 \label{eq:app-particle}
\end{equation}
and its axial trace is $G_H=-D\eta_H/d_H$. Since $H_*(\eta_H)=0$, one has
\begin{equation}
 1-2\eta_HG_H=\frac{L_H}{d_H},\qquad
 \calM_H=\frac{A L_HG_H}{d_H},\qquad
 a_H:= (D_tu_z)_H=\frac{A G_H}{d_H}q_H^{-A-1}.
 \label{eq:app-particle-balance}
\end{equation}
These formulas identify a material particle, rather than a time-dependent spatial maximizer. They will be used to interpret the force partitions below.

\subsection{\texorpdfstring{The $m=4$ pressure--viscosity partition}{The m=4 pressure--viscosity partition}}\label{subsec:app-axis-partition}
Take $G(\eta)=4\eta+j_0$ with $0<j_0\ll1$. Let $\Pi_0$ be the even reference pressure and let $\eta_H$ denote the unique zero of $H_*$. The reference construction has $\cal P_H>0$ and $Z_{*,0}(\eta_H)>0$. It follows from \eqref{eq:app-axis-balance} that pressure supplies more than the net acceleration, while the viscous term is negative at the distinguished particle:
\begin{equation}
 -\partial_zp_0> a_H,\qquad \nu\Delta u_{z,0}+f_z<0
 \quad\text{at }z_H(t),
 \label{eq:app-reference-partition}
\end{equation}
The acceleration $a_H$ in \eqref{eq:app-particle-balance} already includes the similarity factor, so the displayed comparison is between physical acceleration terms.

The pressure can be changed by a small linear tilt while the same $G$ and the same first Cartesian velocity gradient are retained. More precisely, for a prescribed finite $\beta\ge0$, choose
\begin{equation}
 \Pi_\beta(\eta)=\Pi_0(\eta)+c_\beta\eta,\qquad
 c_\beta=\frac{\cal P_H+\beta\calM_H}{d_H-4A\eta_H^2},
 \label{eq:app-beta}
\end{equation}
where the quantities on the right are evaluated for the reference data. The denominator is positive for the small offset under consideration. For a fixed compact range $0\le\beta\le\beta_{\max}<\infty$, choose $j_0$ sufficiently small for the pressure edits. The reference member and the modified family are reconstructed separately with common axis data, as proved in Section~\ref{subsec:app-partition-reconstruction}; they are not connected through a zero of $Z_*(\eta_H)$. No bound uniform as $\beta_{\max}\to\infty$ is asserted. The outer pressure change is realized by the finite-dimensional moment correction described in Appendix~\ref{app:completion}; it is not an independently added pressure field. At the same material particle,
\begin{equation}
 -\partial_zp_\beta=-\beta a_H,\qquad
 \nu\Delta u_{z,\beta}+f_{z,\beta}=(1+\beta)a_H,
 \label{eq:app-beta-balance}
\end{equation}
with $a_H$ again denoting the physical acceleration. Thus the two completed fields have identical prescribed late axis motion and first gradient, but opposite signs of leading viscous work for $\beta>0$.

The comparison does not violate the local energy identity. If $S_u=(\nabla u+\nabla u^T)/2$, then
\begin{equation}
 \nu u\cdot\Delta u=\nabla\cdot(2\nu S_u u)-2\nu|S_u|^2 .
 \label{eq:app-local-energy}
\end{equation}
The paired fields have the same first gradient on the axis and hence the same local dissipation density there. Their viscous transport term differs because their second and higher radial organization differs. The sign change is therefore a redistribution statement, not a positive total viscous-energy source. Taking $\mathcal O$ to record the common late axis velocity and first gradient, this pair lies in one observation fiber while the leading pressure--viscosity partition differs. It therefore disproves identifiability of that partition from these observations in the sense of \eqref{eq:identifiability}.

\subsection{One tilted pressure and signed material motion}\label{subsec:app-axis-signed}
Fix a small nonzero $\gamma>0$ and prescribe the canonical axis pressure trace and axial velocity trace by
\begin{equation}
 \Pi_\gamma=\Pi_{\rm ref}+\gamma\eta,\qquad G_j=4\eta+j,\qquad |j|\le j_* .
 \label{eq:app-signed-data}
\end{equation}
The transport root is the unique solution of
\[
 H_j(\eta)=D\eta+d(4\eta+j)=0,
\]
and satisfies
\begin{equation}
 \eta_j=-\frac{j}{D+4}+O(j^3),\qquad
 G_j(\eta_j)=\frac{D}{D+4}j+O(j^3).
 \label{eq:app-signed-root}
\end{equation}
Writing $K_p=\Pi_{\rm ref}''(0)-4A\Pi_{\rm ref}(0)>0$, the source and pressure coefficients at this root obey
\begin{equation}
 \calM_j=\frac{AD}{D+4}j+O(j^3),\qquad
 \calP_j=-\gamma+\frac{K_p}{D+4}j+O(\gamma j^2+j^3).
 \label{eq:app-signed-coefficients}
\end{equation}
After reducing $j_*$ if necessary, both $\calP_j<0$ and $Z_j=\calP_j-\calM_j\le-\gamma/2$. Equation~\eqref{eq:app-axis-balance} then gives three material behaviors under the same canonical axis pressure trace:
\begin{enumerate}
\item for $j>0$, the particle approaches from below and the positive viscous acceleration acts against the pressure force;
\item for $j=0$, the origin is stationary, while its first velocity gradient and the off-axis swirl still become singular;
\item for $j<0$, the particle approaches from above and accelerates downward under the pressure force, while viscosity brakes it.
\end{enumerate}
At $j=0$, $u(0,t)=D_tu(0,t)=0$ but the complete velocity supremum still diverges through the shrinking positive-swirl region. The distinction is between a spatial supremum and the behavior of any one fixed axis particle.

\subsection{Stretching across the critical rate}\label{subsec:app-axis-stretching}
For the unbiased trace $G=m\eta$ with a fixed nonzero tilt, the late center gradient is
\begin{equation}
 \nabla u(0,t)=
 \begin{pmatrix}
 -m/(2\tau)&-C^{-1}\tau^{-1-h}&0\\
 C^{-1}\tau^{-1-h}&-m/(2\tau)&0\\
 0&0&m/\tau
 \end{pmatrix}.
 \label{eq:app-center-gradient}
\end{equation}
The exact vorticity balance at the stationary center is
\begin{equation}
 \nu\Delta\omega_z+(\nabla\times f)_z
 =\frac{1+h-m}{\tau}\omega_z .
 \label{eq:app-vorticity-balance}
\end{equation}
Proposition~\ref{prop:critical-connection} proves the realizability interval in Theorem~\ref{thm:families}(\ref{item:family-stretching}), with a fixed small positive interval of $h$ and a finite $\Delta_->0$ satisfying
\begin{equation}
 1+h-\Delta_-\le m\le4,
 \label{eq:app-m-range}
\end{equation}
and can be chosen to include $m=1$. For each completed member, \eqref{eq:app-vorticity-balance} gives three regimes for the leading direct viscous supply: positive for $m<1+h$, vanishing for $m=1+h$, and negative for $m>1+h$. At the critical rate the exact statement is $\nu\Delta\omega_z=-(\nabla\times f)_z$ at the center, not an identically vanishing Laplacian for all $t<1$. At $m=1$, stretching supplies the coefficient $1$ of the growth rate $1+h$, while direct diffusion supplies the remaining coefficient $h$.

The quantity in \eqref{eq:app-vorticity-balance} is distinct from the viscous supply of axial strain. If $s=\partial_z u_z$ and
\[
 K_p=\Pi_{\rm ref}''(0)-4A\Pi_{\rm ref}(0)>0,
\]
then
\begin{equation}
 \nu\Delta s+\partial_zf_z=[K_p+m(m+1)]\tau^{-2}>0 .
 \label{eq:app-strain-balance}
\end{equation}
Thus the direct vorticity-diffusion channel can change sign while the leading viscous supply of strain remains positive. At fixed $h$ and $C$, the members share the central vorticity trace $\omega_z=2C^{-1}\tau^{-1-h}$, but \eqref{eq:app-center-gradient} shows that their complete first gradients vary with $m$. This is a different observation fiber from the fixed-axis-gradient comparison above. For $m<1+h$, the analytic core has a short interval of negative azimuthal shear near the origin; the active annulus begins only after positive shear and the strict cone have been recovered.

\subsection{Higher axis geometry and pressure curvature}\label{subsec:app-axis-analytic}
Fix one completed critical member, keep $h$ and $F_*=F_{*,0}$ fixed, and use the common complex neighborhood $\Om$ of Theorem~\ref{thm:families}(\ref{item:family-analytic}). If $g$ and $\psi$ are bounded and holomorphic on $\Om$, real on $[-1,1]$, and satisfy
\begin{equation}
 \norm g_\Om+\norm\psi_\Om<\eps_F,\qquad g(0)=g'(0)=0,
 \label{eq:app-function-neighborhood}
\end{equation}
then
\begin{equation}
 G=(1+h)\eta+g(\eta),\qquad
 \Pi_{\rm ax}=\Pi_{{\rm ax},0}+\psi(\eta)
 \label{eq:app-function-data}
\end{equation}
has the source-dependent realization asserted there. The pressure-restoration step uses three compactly supported swirl modifications on a prepared region where $U=0$. Its moment map controls $\Delta I$, $\Delta S$ and $\Delta C_p$, with the sign convention for $C_p$ defined in \eqref{eq:app-pressure-moment}. This step requires a uniformly nonsingular bump-response matrix, not merely three differently written moment functionals. Solving
\begin{equation}
 \Delta I=0,\qquad \Delta S=0,\qquad \Delta C_p=-\psi
 \label{eq:app-pressure-matching}
\end{equation}
leaves the prepared exterior unchanged and restores the canonical pressure exactly. Here $\Delta$ denotes the finite increment of this pressure-restoration step, after the other target moments have been restored. For a simultaneous change of $G$ and pressure, the accumulated five-moment defect, including that generated by the core change, must instead be restored using \eqref{eq:app-five-map}; the zero targets in \eqref{eq:app-pressure-matching} do not discard earlier defects. The quadratic terms in $\Delta S$ and $\Delta C_p$ are retained. The actual rank and parameter-uniform moment inverse are proved in Appendix~\ref{app:completion}; Section~\ref{subsec:app-local-lifting} supplies the core contraction, flat attachment, finite connection and mean-preserving modulation for the entire function ball.

Because $g(0)=g'(0)=0$ and $F_*$ is fixed, the central velocity and first gradient are unchanged. The assertion for every pair in the stated ball therefore gives \eqref{eq:trace-surjectivity}. It does not preserve the entire axis first-gradient trace for arbitrary $g$, nor does it provide a continuous selection of the final oscillatory fields.

An explicit closed slice is
\begin{equation}
 g(\eta)=\kappa\eta^3,\qquad \psi(\eta)=\frac{\rho}{2}\eta^2,
 \label{eq:app-explicit-slice}
\end{equation}
with $|\kappa|$ and $|\rho|$ below the finite reconstruction radius. It preserves the center gradient and the critical vorticity balance. Write $\delta_{\rm fam}F:=F_{\kappa,\rho}-F_{0,0}$ for a finite difference between family members, to distinguish it from the spatial Laplacian $\Delta$. Then
\begin{equation}
 \delta_{\rm fam}(\partial_z^3u_z)(0,t)=6\kappa\tau^{-2+2h},\qquad
 \delta_{\rm fam}(\nu\Delta s+\partial_zf_z)=\rho\tau^{-2}.
 \label{eq:app-slice-observables}
\end{equation}
The first parameter changes higher axis geometry and the second changes pressure curvature and strain supply. Neither is a change of the center pressure value or its linear tilt.

\subsection{The analytic core and its parameter estimates}
\label{subsec:app-core-contraction}
We give the reconstruction estimates used in the preceding assertions.
Only the quartic regularizer and the original prepared exterior of slope
four are needed. No assertion uniform in regularizer order, zero tilt or
$h=0$ is used. All constructions in the next three subsections are in
viscosity-one coordinates.

Write $\varphi=CE/\sqrt{2X}$, $\varphi_*=CF_*$, and let
$\zeta=\Lambda^{-1}(\log\varphi_*)'$. For prescribed $G,\Pi_{\rm ax}$
put
\[
 H_*=D\eta+dG,\quad W_*=1-2D\eta G-dG',\quad
 \chi=-H_*\zeta/L,
\]
with $Z_*$ as in \eqref{eq:app-axis-def}. On setting $Y=\Lambda X$, $D_Y=Y\partial_Y$,
$\varphi=\varphi_*\Phi$, $U=G+u/\Lambda$, introduce
\[
 g_a=\varphi_*/C,\quad
 p_1^{\rm rad}(Y)=\int_0^Y g_a^2\Phi^2\,\dd Y',\quad
 \overline u=Y^{-1}\int_0^Y u\,\dd Y',\quad
 B=-2D\eta\overline u-d\partial_\eta\overline u.
\]
Here $p_1^{\rm rad}$ is a radial pressure increment, not the stress
coordinate $p_1$ of \eqref{eq:cone-coordinates}. Then
$W=W_*+B/\Lambda$, $H_c=H_*+du/\Lambda$ and
$\Pi=\Pi_{\rm ax}+p_1^{\rm rad}/\Lambda$.
Direct substitution in the stress-free profile equations gives
\begin{equation}
 \begin{aligned}
 2(Y\Phi_{YY}+2\Phi_Y)&=-\chi\Phi+\Lambda^{-1}\mathcal R_1,\\
 2(Yu_{YY}+u_Y)&=-Z_*/L+\Lambda^{-1}\mathcal R_2,\\
 L\mathcal R_1&=[W+h(1-2\eta U)+du\zeta]\Phi
                     +W D_Y\Phi+H_c\Phi_\eta,\\
 L\mathcal R_2&=[A(1-4\eta G)+dG']u-2A\eta u^2/\Lambda
                 +W D_Yu+H_*u_\eta+du u_\eta/\Lambda\\
              &\hspace{1em}-4A\eta p_1^{\rm rad}
                   +d(p_1^{\rm rad})_\eta-2\eta D_Yp_1^{\rm rad}.
 \end{aligned}
 \label{eq:core-normalized-system}
\end{equation}
The equalities, including the parameter derivatives, are re-solved for
each input. There is no independent resetting of the pressure or radial
mean.

We use the coefficient-space estimates of
\cite[Lemma~B.1]{OpenAI2026}. For completeness, the norm for
$F=\sum_{\alpha\ge0}F_\alpha(\eta)Y^\alpha$ is
\begin{equation}
 a_{\alpha\beta}=
 \frac{20^{-\alpha}\varrho^{-\beta}\beta!
             \binom{\alpha+\beta}{\beta}}
      {(\alpha+1)^2(\beta+1)^2},\qquad
 \|F\|_{\mathscr B_\varrho}
 =\sup_{\alpha,\beta,\eta\in I}
       \frac{|\partial_\eta^\beta F_\alpha(\eta)|}{a_{\alpha\beta}}.
 \label{eq:core-coefficient-norm}
\end{equation}
Here $I$ is a fixed slightly enlarged real interval. Radial averaging,
multiplication and integration have bounded forms in this space. If
$\mathcal J_\ell$ inverts $Yv_{YY}+\ell v_Y$ with $v(0)=0$, then
\[
 (\mathcal J_\ell F)_{\alpha+1}
 =\frac{F_\alpha}{(\alpha+1)(\alpha+\ell)},\qquad \ell=1,2.
\]
In particular $\mathcal J_\ell[(\partial_\eta F)D_YG]$ is bounded by
$C_\varrho\|F\|\|G\|$: the extra radial degree compensates the parameter
derivative. This is why ordinary fixed-order formal Taylor series are
not being used as an existence argument.

Let $\mathcal T=\mathcal J_2\mathsf M_\chi/2$, where $\mathsf M_\chi$
is multiplication by $\chi$. If its operator norm is bounded by
$C_\chi$, the iterates satisfy
\begin{equation}
 \|\mathcal T^k\|\le
 \frac{(40C_\chi)^k}{k!(k+1)!}.
 \label{eq:core-volterra-inverse}
\end{equation}
Thus $1+\mathcal T$ has a bounded inverse without requiring $\chi$ to
be small. With
\[
 f_0(z)=\sum_{\alpha\ge0}\frac{(-z/2)^\alpha}{\alpha!(\alpha+1)!},
 \quad \Phi_0=f_0(Y\chi),\quad u_0=-YZ_*/(2L),
\]
solve \eqref{eq:core-normalized-system} by
\begin{equation}
 (\Phi,u)=(\Phi_0,u_0)+\frac1{2\Lambda}
 \bigl((1+\mathcal T)^{-1}\mathcal J_2\mathcal R_1,
                         \mathcal J_1\mathcal R_2\bigr).
 \label{eq:core-fixed-point}
\end{equation}
On any fixed ball around $(\Phi_0,u_0)$, the integrated remainders and
their Lipschitz constants are bounded by the coefficient bounds just
listed. First choose $\Lambda$ so that this map sends the ball into a
strictly smaller concentric ball and has Lipschitz constant at most
$1/2$. Then choose $C$ to bound $g_a$ on the \emph{complex}
coefficient neighborhood, not just on the real interval. These choices
are uniform for compact sets of the fixed coefficient data. The fixed
point satisfies, in every required fixed derivative norm on
$0\le Y\le4.1$,
\begin{equation}
 \Phi=f_0(Y\chi)+O(\Lambda^{-1}),\qquad
 u=-YZ_*/(2L)+O(\Lambda^{-1}).
 \label{eq:core-uniform-comparison}
\end{equation}
For the quartic data below $0\le\chi\le1$, and the alternating-series
bound gives $f_0(z)>.265$ on $[0,4.1]$. Thus the reconstructed core has
positive $F$, is regular at the Cartesian axis and solves $T_0=0$
exactly. Positivity of the shear at every inner radius is not required.

\subsection{Global connection across the critical stretching rate}
\label{subsec:app-critical-connection}
The center identities alone do not give the global connection. We provide
it here for the fixed quartic regularizer and the original exterior of
slope four. Throughout this subsection $C_j$ denotes a finite constant
in the indicated fixed parameter-derivative norm. A constant may depend
on parameters already chosen, but not on parameters explicitly chosen
later. Constants in normalized pressure and moment estimates are uniform
in the subsequently small positive $h$.

\begin{proposition}[A completed cross-critical interval]
\label{prop:critical-connection}
Under Assumption~\ref{ass:completion}, there are finite choices of
$K=P_*^2$, $\Lambda$, $C$, a fixed $\gamma>0$, and
$0<h_-<h_+<1/100$ such that, for some $\eps_->0$ and
$\Delta_-=\eps_-K/\Lambda$, every
\[
 h\in[h_-,h_+],\qquad 1+h-\Delta_-\le m\le4
\]
has a globally matched leading profile with $G=m\eta$.
At fixed $h$, the axis pressure is
$\Pi_{{\rm ref},h}+\gamma\eta$ for all these profiles, and their
center normalization is the same $C^{-1}$. The positive interval of
$h$ can be chosen with $h_+<\Delta_-/4$, so it includes $m=1$ for every
such $h$.
\end{proposition}
\begin{proof}
We give the construction and the order of its finite choices explicitly.

\paragraph{1. Preparation and pressure scale.}
Use the outer schedule of \cite[Appendix~A.2]{OpenAI2026}, with
$U=4\eta$, $E=P_*f(X/X_R)^{1/10}$ on its initial reference interval,
where $f=(1+\eta^2)^{-1}$. Fix its cutoff shapes and $M_d,T_d$, then
$P_*$ large and $\lambda>0$ small. Its pressure before the heat
compensation can be written
\[
 \Pi_{{\rm ref},h}=-\frac12\int
       c_h(y)^2f^{2\theta_h(y)}\,\dd y,\qquad 0\le\theta_h\le1.
\]
The initial reference contributes $-5Kf^2/2$. After a fixed transition,
$l=\partial_y\log(\sqrt{2X}E)\le0$, so
$(E^2)_y=(2l-1)E^2\le-E^2$. On one fixed complex neighborhood avoiding
the zeros and poles of $f$, integration consequently gives
\begin{equation}
 -\Pi_{{\rm ref},h}\ge\tfrac52Kf^2,\qquad
 \Pi_{{\rm ref},h}'/\eta\ge0,\qquad
 \|\Pi_{{\rm ref},h}\|_{C^j}\le C_jK.
 \label{eq:connection-pressure-preparation}
\end{equation}
At zero the quotient is interpreted by continuity. The same bound holds
in a smaller fixed holomorphic neighborhood. It is uniform as the later
positive $h$ decreases: the long late tail is exponentially suppressed
in this pressure integral. The source's angular and heat compensations
preserve the pressure increment exactly. The total $S$-moment is still
computed only for $h>0$; the suppression by $h^8$ before the slow tail
in \cite[Equations~(A.17)--(A.18)]{OpenAI2026} leaves an $O(h^7)$
tail contribution. We do not assert an $h=0$ completed profile.

Use the three-bump pressure edit of Lemma~\ref{lem:prepared-rank} on a
fixed pre-tail power patch to add $\gamma\eta$, for one sufficiently
small fixed $\gamma>0$. Its inverse and pre-tail cone margins are
uniform in the later $h$. Beyond this edit the actual terminal
preparation is unchanged, so its possibly shrinking late-edge margins
are not used to choose $\gamma$.

\paragraph{2. Actual core data and the positive-radius exit.}
Put $\delta=m-1-h$ and, on a preliminary compact range
$.9\le m\le4$, set
\begin{equation}
 H_m=\eta(D+md),\quad
 \chi=\frac{H_m^4}{H_m^4+\sigma_*^4},\quad
 \zeta=-\frac{LH_m^3}{H_m^4+\sigma_*^4},\quad
 \varphi_*=\exp\left(\Lambda\int_0^\eta\zeta\,\dd s\right).
 \label{eq:connection-quartic-data}
\end{equation}
The axial source computed from the actual pressure is
\begin{equation}
 \begin{split}
 Z&=-\eta B_{m,h}+\gamma[(4A+1)\eta^2-1],\\
 B_{m,h}&=m(m+1)-2(1+h)m^2\eta^2
           +d\Pi_{{\rm ref},h}'/\eta-4A\Pi_{{\rm ref},h}.
 \end{split}
 \label{eq:connection-general-Z}
\end{equation}
In particular the cubic coefficient is $2(1+h)m^2$, not $2m^3$
away from criticality. Equation~\eqref{eq:connection-pressure-preparation}
gives $cK\le B_{m,h}\le CK$ and bounded fixed derivatives after division
by $K$. There is one simple root $z$ of $Z$ with
$-C\gamma/K\le z\le-c\gamma/K$. To see uniqueness, outside a
sufficiently large multiple of $\gamma/K$ the $-\eta B_{m,h}$ term
dominates; inside it $Z'\le-cK$. Choose $\sigma_*$ small enough that
$\chi>.99$ wherever $|Z|$ is below a fixed positive threshold.
This is possible uniformly on the preliminary compact $(m,h)$ set,
since $H_m$ vanishes only at zero and $Z(0)=-\gamma$.
The denominators stay nonzero on a common smaller complex neighborhood.

Define the exact angular source
\begin{equation}
 S_q=-Wl-h(1-2\eta U)-(D\eta+dU)(\log E)_\eta,
 \quad l=1-a/2.
 \label{eq:connection-angular-source}
\end{equation}
Differentiating the actual cumulative definition of $Q$ gives
\begin{equation}
 Q_y+(1+l)Q=S_q,\qquad
 (p_1)_y=XS_q/L-lp_1.
 \label{eq:connection-source-evolution}
\end{equation}
At the axis $S_q(0,\eta)=\delta+\Lambda L\chi$.
The core solve of Section~\ref{subsec:app-core-contraction} is available
on $0\le Y=\Lambda X\le4.1$. It remains to retain one order of
information beyond \eqref{eq:core-uniform-comparison} near $\eta=0$.
Writing $J=-Z'(0)\asymp K$, its exact center coefficients are
\begin{equation}
 \begin{gathered}
 \varphi_X/\varphi=-\delta/4,\qquad U_X=\gamma/2,\qquad
 U_{X\eta}=J/2,\\
 (S_q)_X=[J-(m-1)\delta]/4,\qquad
 \varphi_{XX}/\varphi=-[J-(m-1)\delta-\delta^2]/24,\\
 a(X,0)=\delta X/2+
       [J/12-(m-1)\delta/12+\delta^2/24]X^2+O(X^3).
 \end{gathered}
 \label{eq:connection-signed-jets}
\end{equation}
These identities also exhibit the negative inner shear when $\delta<0$.
They do not alone justify a uniform exit radius.

Here is the requisite remainder estimate. Let
$\rho=\Phi-f_0(Y\chi)$. The radial traces are
$\rho(0,\eta)=0$ and $\rho_Y(0,\eta)=-\delta/(4\Lambda L)$.
At $\eta=0$ the derivatives of $\chi$ through order three and of
$\zeta$ through order two vanish. Differentiate
\eqref{eq:core-normalized-system} through two parameter derivatives
there, subtract the displayed linear radial trace, and apply the same
radial inverses. The remainder starts at radial degree two; its source
has the extra factor $\Lambda^{-1}$. The terms $H_m\rho_\eta$ do not
introduce the next uncontrolled jet, since $H_m(0)=0$.
The weighted mixed estimate controls the other products. This gives,
for $0\le j\le2$,
\[
 |\partial_\eta^j\rho(Y,0)|+|D_Y\partial_\eta^j\rho(Y,0)|
 \le C\bigl(|\delta|Y/\Lambda+Y^2/\Lambda^2\bigr).
\]
Taylor's formula in $\eta$, with the next derivatives bounded by
\eqref{eq:core-coefficient-norm}, then yields on a fixed small
neighborhood of zero
\begin{equation}
 \begin{aligned}
 |\rho|+|D_Y\rho|&\le C\left[
    |\delta|Y/\Lambda+Y^2(|\eta|^3/\Lambda+\Lambda^{-2})\right],\\
 |\rho_\eta|&\le C\left[
    |\delta|Y/\Lambda+Y^2(\eta^2/\Lambda+\Lambda^{-2})\right].
 \end{aligned}
 \label{eq:connection-refined-remainder}
\end{equation}
No vanishing noncritical radial trace has been assumed in this estimate.

Put $a_1=-Z/(2L)$. The axial equation gives
$U=m\eta+Xa_1+O_{C^j}(X/\Lambda)$ and
$\mu=m\eta+Xa_1/2+O_{C^j}(X/\Lambda)$, whence
\[
 -W-h(1-2\eta U)=\delta+Xb_1+O(X/\Lambda),\qquad
 b_1=da_1'/2+A\eta a_1,\quad b_1(0)=J/4.
\]
The most substantial phase error in
\eqref{eq:connection-angular-source} is controlled by
\[
 C\Lambda X\chi^{3/4}
       \le\epsilon\Lambda\chi+C_\epsilon\Lambda X^4.
\]
The other phase errors from \eqref{eq:connection-refined-remainder}
are bounded by $C\Lambda X^2\chi^{3/4}$,
$C\Lambda X^3\chi^{1/2}$ and $CX^2$, whose Young remainders are
$C_\epsilon\Lambda X^8$, $C_\epsilon\Lambda X^6$ and $CX^2$.
On $X\le4.1/\Lambda$ all these remainders are $o(X)$.
The terms $C|\delta|X$ are a small fraction of $|\delta|$.
Near zero the positive $b_1$ supplies the remaining $KX$ term; away
from that neighborhood $\chi$ has a positive lower bound. Therefore,
with $\delta_\pm=\max(\pm\delta,0)$,
\begin{equation}
 \begin{aligned}
 S_q&\ge c_1\delta_+-C_1\delta_-+c_2KX+c_3\Lambda\chi,\\
 p_1&\ge c_4\delta_+X-C_2\delta_-X+c_5KX^2+c_6\Lambda X\chi.
 \end{aligned}
 \label{eq:connection-signed-source-bound}
\end{equation}
The second inequality follows by the positive integrating kernel in
\eqref{eq:connection-source-evolution}, using bounded positive $\Phi$.
Choose $\eps_->0$ small and restrict
$\delta_-\le\eps_-K/\Lambda$. At $X_0=4/\Lambda$ and on its small
natural right collar, the radial positive term dominates the adverse
one. Hence
\begin{equation}
 p_1\ge c\delta_+/\Lambda+cK/\Lambda^2+c\chi>0,
 \qquad p_1+p_2^2/p_1>2+c_{\rm ex}.
 \label{eq:connection-actual-exit}
\end{equation}
For the second statement, $\chi>.99$ gives the source Bessel lower
bound $p_1>2.3$. On the complement $|N|$ stays away from zero and
$p_2=XN/(LE)$ grows in absolute value with the subsequently large $C$.
In the stress-free core $p_1=a$, $p_2=-b$, so this is precisely the
positive-shear, strict exit needed at the active edge.

\paragraph{3. Reference continuation and flat stress activation.}
Freeze the natural logarithmic $\varphi$ slope and $U$ slope to zero on
a sufficiently short part of this right collar, using the reference
cutoffs of \cite[Equation~(B.22)]{OpenAI2026}. At the end $X_f$ of the
cutoff let $G_r$ be its constant axial trace. The actual radial mean is
\begin{equation}
 \begin{split}
 G_r&=m\eta+4a_1/\Lambda+O_{C^j}(\Lambda^{-2}),\\
 \mu(X_0)&=m\eta+2a_1/\Lambda+O_{C^j}(\Lambda^{-2}),\\
 \mu_r(X)&=G_r+(X_f/X)[\mu_r(X_f)-G_r].
 \end{split}
 \label{eq:connection-retained-increment}
\end{equation}
Choose the cutoff errors smaller than $\Lambda^{-2}$ in the finitely
many required norms. The extra derivative in this mean ranges between
$J/\Lambda$ and $2J/\Lambda$ near zero; it must not be discarded by
replacing $G_r$ or $\mu_r$ with $m\eta$.
Using \eqref{eq:connection-refined-remainder} in the frozen angular
phase gives
$C\chi^{3/4}\le\epsilon\Lambda\chi+C_\epsilon\Lambda^{-3}$.
The other low-jet errors are of smaller order than $K/\Lambda$;
$C|\delta|/\Lambda$ is absorbed by $\delta_+$ or the chosen thin
negative strip. Thus the actual reference obeys
\begin{equation}
 S_{q,r}\ge c\delta_++cK/\Lambda+c\Lambda\chi,\qquad l_r\le1.
 \label{eq:connection-reference-floor}
\end{equation}
Pressure changes on a selected finite interval are $o_C(1)$ and are
controlled by choosing $C$ later. The scalar comparison gives
\[
 p_{1,r}(X)\ge p_{1,r}(X_0)X_0/X+
       \frac{c}{2\Lambda}(X-X_0^2/X).
\]
Take a sufficiently large finite $R_b$ and continue to
$X_b^{\rm ref}=R_b\Lambda$, $X_i=2R_b\Lambda$, so that $p_{1,r}>10$
near the end. The exit alternative in
\eqref{eq:connection-actual-exit} persists on the whole finite reference
interval: on the large-$\chi$ part by its positive angular term, and on
the complement by $N_r=Z+O(\Lambda^{-1})+o_C(1)$ and the nonzero-$Z$
alternative.

With $y=\log(X/X_0)$, use the actual shear reduction
\begin{equation}
 e_a=(1-\kappa_0)\sigma(y/t_1),\quad \kappa=1-e_a,
 \qquad a=\kappa p_{1,r},\qquad
 U_y=-\kappa XN_r/(2L).
 \label{eq:connection-flat-activation}
\end{equation}
Here $\sigma$ is the source flat step and the reference is unchanged
for $X\le X_0$. The exact differences satisfy
$(\log\varphi-\log\varphi_r)_y=e_a p_{1,r}/2$ and
$(U-U_r)_y=e_a XN_r/(2L)$.
Their integrals and the actual moment formulas give
$p-p_r=O(ye_a)$ and $E_r/E=1+O(ye_a)$.
Consequently $P_c-v\ge ce_a$, $P_c>2+c$ and
$(v-2)_+J_c^2<2(P_c-v)^2$ for small $t_1$.
After dividing by the prescribed flat factor,
$T_0=e_aB_0$ with
$B_0(0,\eta)=F(X_0,\eta)p_r(X_0,\eta)\ne0$.
The flat-integral calculation in
\cite[Lemma~A.9 and Equations~(B.27)--(B.31)]{OpenAI2026} supplies
smooth normalized directions and every fixed weighted derivative bound.
This verifies the edge condition, not just a strict inequality away
from a vanishing stress.

Choose the small constant multiplier $\kappa_0$ and the remaining
cutoff widths to continue on this finite interval, then turn off $U_y$
and set $a=.8$. The comparisons above and the shear-reduction argument
of \cite[Proposition~B.5]{OpenAI2026} give a genuine continuation with
a first strict pulse-cone collar and the strict relaxed cone thereafter.
At $X_i$ one has $p_1>8$, $U=G_i$, $U_y=0$, and
\begin{equation}
 \|G_i-m\eta\|_{C^k}\le\delta_{\rm core},\qquad
 \|(G_i)_{\rm even}\|_{C^k}\le\gamma\delta_{\rm core},\qquad
 \delta_{\rm core}=C_{\rm core}/\Lambda\ll K^{-1}.
 \label{eq:connection-core-discrepancy}
\end{equation}
The constant depends on the earlier pressure and regularizer data.
Reflection equivariance and smooth dependence of the core on $\gamma$
give the even-part estimate. An auxiliary core at $\gamma=0$ is used
only for this parity estimate, not as a completed exit.

\paragraph{4. Shape adjustment at a common normalization.}
Set $\ell_i=\log(CE(X_i,\eta))$. Its required parameter norms have a
bound independent of subsequently large $C$ and small transition
widths, by integration of the bounded reference slopes. Choose a fixed
finite length, depending on that bound, and prescribe
\begin{equation}
 \log E=-\log C+y_i/10+(1-\theta)\ell_i+\theta\log f,
 \qquad U=G_i,\quad y_i=\log(X/X_i).
 \label{eq:connection-common-C-shape}
\end{equation}
A sufficiently long smooth transition keeps $.55\le l\le.65$.
At its end $E=C^{-1}f(X/X_i)^{1/10}$ exactly; an $m$-dependent scalar
has not been absorbed into $C$. The same phase estimates as in
\eqref{eq:connection-reference-floor}, with the retained actual radial
increment, give a constant part
$l\delta-(1-l)h+cK/\Lambda$ and the positive phase terms
$c(1-\theta)\Lambda\chi+c\theta\eta^2$.
Choose the later $h_+$ with $h_+\le c_hK/\Lambda$ and reduce
$\eps_-$ if necessary. Then
\begin{equation}
 S_q\ge c\delta_++cK/\Lambda+
            c(1-\theta)\Lambda\chi+c\theta\eta^2.
 \label{eq:connection-shape-floor}
\end{equation}
This keeps $p_1>6$. Continue the power profile until its amplitude is
moderate, and change $l$ smoothly to $1/2$ so that the constant
amplitude on the next hold is a fixed $e_0$, which may be one.
The scalar entry and exit schedules are fixed; their pressure
increments are included below.

\paragraph{5. Axial steering with the actual moving zero.}
On a fixed finite hold prescribe
\begin{equation}
 E=e_0f,\qquad a=1,\qquad
 U=G_i+r(y)(\eta-z),\quad
 r(0)=0,\quad r(T)=4-m,\quad 0\le r'\le R_0<.05.
 \label{eq:connection-axial-steering}
\end{equation}
One common smooth schedule of sufficiently large finite length $T$
works for $0\le4-m<4$. All its endpoints are flat. Set $V=\eta-z$,
$\overline r_y=r-\overline r$, and
$\overline{r^2}_y=r^2-\overline{r^2}$ with zero initial values.
The exact quadratures are
\begin{equation}
 \begin{aligned}
 \mu&=G_i+\overline r V+e^{-y}\rho_{\mu,0},\\
 \sigma&=G_i^2+2G_iV\overline r+V^2\overline{r^2}
             -e_0^2c_\sigma(y)f^2+e^{-y}\rho_{\sigma,0},\\
 \Pi&=\Pi_{{\rm ref},h}+\gamma\eta+
                e_0^2(c_P+y/2)f^2+\rho_{\Pi,0},\qquad
 (c_\sigma)'=1/2-c_\sigma.
 \end{aligned}
 \label{eq:connection-actual-hold-moments}
\end{equation}
The $\rho$ terms are the actual entry-history errors, made small by
large $C$; they are not set to zero. The linear growth of pressure on
the hold and the $\overline{r^2}$ term are retained.

For increasing bounded $r$ with the indicated derivative bound,
$\overline r\ge c r^2$. Indeed integrate $r$ backwards over length
$\min(1,r/(2R_0))$, on which it is at least half its final value.
Substitution in \eqref{eq:connection-angular-source} shows that the new
positive mean-gradient term is $l\overline r L$. The indefinite bias
cost is at most $C(r+\overline r)|z||\eta|$ plus the already small
core bias. Complete the square using $|z|=O(\gamma/K)$ and the bound
on $\overline r$ to obtain
\begin{equation}
 S_q\ge cK/\Lambda+c\overline r+c\eta^2,\qquad
 Q\ge cK/\Lambda+c\eta^2,\qquad p_1>6.
 \label{eq:connection-hold-Q}
\end{equation}
The coefficient of $\eta^2$ can be fixed independently of large $K$:
the $f$ gradient contributes $2D\eta^2/(1+\eta^2)$, with $D\ge.49$.
Near zero the core increment supplies the positive floor; away from
zero its small remaining error is absorbed into this fixed margin.

Insert all of \eqref{eq:connection-actual-hold-moments} in the actual
$N$ formula. Odd terms from the bounded geometric velocities and the
hold pressure are bounded by $C(1+e_0^2T)|\eta|$; the even axial
parts are $O(\gamma/K+\gamma\delta_{\rm core})$.
Thus, with corresponding fixed derivative bounds,
\begin{equation}
 |N-Z|\le C(1+e_0^2T)|\eta|+C\gamma/K
       +C\gamma\delta_{\rm core}+o_C(1)\gamma.
 \label{eq:connection-actual-N-error}
\end{equation}
The constants here are uniform in large $K$ after
\eqref{eq:connection-core-discrepancy}. Outside a fixed multiple of
$\gamma/K$ the $-\eta B_{m,h}$ term dominates. Inside it
$N_\eta\le-cK$, so $N$ has exactly one root $z_N(y)$, with
\begin{equation}
 |z_N-z|\le C\gamma(1+e_0^2T)/K^2,
 \qquad |z_N|\asymp\gamma/K.
 \label{eq:connection-actual-zero}
\end{equation}
This is a bound for the root of the actual residual, not a substitution
of the frozen reference root. It follows by evaluating
\eqref{eq:connection-actual-N-error} at $z$ and using the mean value
theorem; choose the history errors smaller than $\gamma/K$ first.

Here $b=2r'(\eta-z)/(e_0f)$, and $U_yN$ can be positive only between
$z$ and $z_N$. On that interval
$(U_yN)_+\le CK R_0|z_N-z|^2$, whereas
$Q\ge c\gamma^2/K^2$. Hence
\begin{equation}
 \left(\frac{bN}{EQ}\right)_+
 \le \frac{CR_0(1+e_0^2T)^2}{e_0^2K}<\frac1{10}.
 \label{eq:connection-steering-cone}
\end{equation}
This last smallness is imposed by the earlier choice of large $P_*$.
Also $|b|<1/2$, so $v=1+b^2<2$ and
$P_c=p_1(1-bN/(EQ))>.9p_1>2$. This is the strict relaxed cone of
\cite[Lemma~4.5]{OpenAI2026}. At the end,
$U=4\eta-(4-m)z+O_{C^k}(\delta_{\rm core})$.
Keep $r$ constant for a fixed further interval to recover a fixed
positive source bound while the radial mean relaxes, then return to
$l=.6$. Include these fixed intervals in the constant $T+1$ below.

\paragraph{6. Exact final moments and the pulse cone.}
Define $X_R$ from the actual resumed power coefficient so that
$E=P_*f(X/X_R)^{1/10}$. The normalized end radius of the hold is
$x_{\rm end}\asymp(e_0/P_*)^{10}$, which is small because of $P_*$,
not because $C$ is large. Write hats for normalization of $M,S$ by
$X_R$ and of $I,J$ by $X_R^{3/2}$. The five scaled rows are
\[
 (\widehat M,\widehat I/P_*,\widehat J/P_*,
       \widehat S/P_*^2,C_p/P_*^2).
\]
Restore the velocity trace to the actual prepared reference on a fixed
subinterval of $e^{-8}<X/X_R<e^{-5}$ and use the remaining ordered
bump supports there. Lemma~\ref{lem:prepared-rank}, with $u_c=4\eta$
and $\alpha=1/10$, gives a uniformly invertible scaled moment map.
The total discrepancy entering this map has, for each required fixed
$k$, the bound
\begin{equation}
 \|\mathrm{defect}\|_{C^k}\le C_k\left[
   \gamma/K+\delta_{\rm core}+(e_0/P_*)^{10}
        +e_0^2(T+1)/P_*^2\right]+o_C(1).
 \label{eq:connection-full-moment-defect}
\end{equation}
The first two terms come from the small residual axial bias and core
error; the third bounds the differing inner radial inventories; the
fourth is the retained pressure accumulated during the ramp, hold and
fixed switches. Earlier finite-radius contributions vanish after this
normalization as $C$ grows. In particular the hold pressure is not
incorrectly included in $o_C(1)$.

In these scaled coordinates $E/P_*$ is bounded below, $N/P_*^2$ is
bounded, and an edit of size $\epsilon$ has
$b=O(\epsilon/P_*)$. Thus $bN/E=O(\epsilon)$, $a$ stays near $.8$
and the strict relaxed cone has an input tolerance independent of
large $P_*$. Choose the parameters so that
\eqref{eq:connection-full-moment-defect} lies within the nonlinear
inverse radius. Solving the \emph{full} moment equations gives exact
pressure and all four cumulative matches. Beyond the restoration
interval, \cite[Lemma~4.4]{OpenAI2026} recovers the actual prepared
exterior and its reserved patches and heat collar.

Apply the source radial modulation
\cite[Lemma~4.11 and Appendix~C]{OpenAI2026} to the remaining relaxed
joining region, protected away from the analytic first collar, all
reserved patches and the terminal collar. Its periodic shear has the
required exact mean; choose a sufficiently large finite frequency and
restore its small moment errors as in Lemma~\ref{lem:prepared-rank}.
The result has the strict pulse cone on the full active annulus, with
the unchanged flat-edge weights. The profile conditions of the source
completion are now met. Negative shear, when present, was confined to
the preceding stress-free core; the positive-order inner equations
use regular radial operators, not $a^{-1}$, and are unaffected by this
inner sign change.

\paragraph{7. Noncircular choice of parameters.}
Fix the exterior cutoff shapes and a common steering schedule $T,R_0,e_0$
first. Choose $P_*$ large enough for
\eqref{eq:connection-steering-cone}, the scaled moment tolerance and
pressure dominance, then choose the fixed positive $\lambda$.
Choose a small positive tilt $\gamma$ using its pre-tail edit, then
$\sigma_*$ and the common analytic neighborhood. Next choose
$\eps_-$ and a sufficiently large $\Lambda$ for the core, reference,
shape and discrepancy estimates, including
$\delta_{\rm core}\ll K^{-1}$; this fixes
$\Delta_-=\eps_-K/\Lambda>0$.
Only now choose a nondegenerate compact positive interval of $h$ with
\[
 h_+<\min\{h_{\rm prep},c_hK/\Lambda,\Delta_-/4,1/100\}.
\]
The pressure estimate \eqref{eq:connection-pressure-preparation} is
what permits this order. On this positive compact interval every late
preparation and heat-collar bound is finite. Finally choose common
large $C$, small activation widths, and the finite modulation and
moment-restoration parameters. All choices precede the subsequent
infinite correction and physical singular limit. This proves the
proposition and Theorem~\ref{thm:families}(\ref{item:family-stretching}).
\end{proof}

\subsection{A local lifting of arbitrary small analytic traces}
\label{subsec:app-local-lifting}
We now prove the function-valued assertion in
Theorem~\ref{thm:families}(\ref{item:family-analytic}), including the
finite equalities in its completion. Fix one critical member constructed
in Section~\ref{subsec:app-critical-connection}, with all its finite
choices. Choose bounded complex neighborhoods
$\Om_2\Subset\Om_1\Subset\Om$ of $[-1,1]$ on which its coefficient
functions and the inverses of the prepared moment blocks are holomorphic.
Work on $\Om_1$ and $\Om_2$, while measuring prescribed $g,\psi$ in
$H^\infty(\Om)$. The domains and all constants in this subsection are
fixed before the inputs are varied. The reconstruction argument also
applies to any other fixed member with the same strict finite margins;
criticality and the two constraints on $g$ are used only for the stated
central-kinematics conclusion.

\paragraph{Actual pressure and core.}
First use the three azimuthal bumps of
Lemma~\ref{lem:prepared-rank} to solve
$(\Delta I,\Delta S,\Delta C_p)=(0,0,-\psi)$ on the reserved $U=0$
patch. The exact quadratic moment map, not its linearization alone, is
inverted by \eqref{eq:app-analytic-moment-lift}. This produces an actual
prepared exterior with axis pressure $\Pi_{{\rm ax},0}+\psi$; beyond
the edit it coincides with the previous exterior. Shrinking the input
ball preserves pressure negativity and $|\psi'(0)|<|\gamma|/2$.

Keep $\varphi_*$, $\Lambda$ and $C$ fixed and replace $G_0$ by $G_0+g$
in \eqref{eq:core-normalized-system}. Then
\[
 H_*=H_{*,0}+dg,\qquad
 \chi=\chi_0-dg\zeta/L.
\]
A bounded holomorphic coefficient on $\Om$ embeds continuously in
\eqref{eq:core-coefficient-norm} with a smaller fixed $\varrho$.
Indeed Cauchy's inequality bounds its $\beta$th derivative by
$\beta!r^{-\beta}\|g\|_\Om$, and
$\sup_\beta(\beta+1)^2(\varrho/r)^\beta<\infty$ for
$\varrho<r$. The same argument with one further Cauchy loss controls
$g'$ and $\psi'$. The integrated nonlinear terms are controlled by the
mixed estimate preceding \eqref{eq:core-volterra-inverse}, so no
unbounded parameter-differentiation operator is being inverted in an
ordinary $H^\infty$ norm.

Put $V_0=(1+\mathcal T_0)^{-1}$. For sufficiently small $\|g\|_\Om$,
\[
 \|V_0(\mathcal T-\mathcal T_0)\|<1/2.
\]
The Neumann inverse retains a bound $2\|V_0\|$. All other terms of
\eqref{eq:core-fixed-point} vary continuously in the same integrated
coefficient norms. The strict ball-invariance and contraction margins
therefore persist, giving an actual core with the prescribed traces,
positive $F$ and continuous dependence in each of the finite certificate
norms. This is a nonlinear solve on one radial interval, not an inference
from a formal axis expansion. The positive-radius exit has strict
margins in $p_1$ and $p_1+p_2^2/p_1-2$, and these persist as well.

\paragraph{Flat attachment and exact global matching.}
Replay the fixed reference cutoffs, shear reduction, shape transition
and steering functions with the new core. Their coefficients are
computed from the new actual velocity, pressure and quadratures
\eqref{eq:app-forward}--\eqref{eq:app-N}. These are finite integrations,
products and parameter derivatives with bounded denominators; their
required $C^k$ norms vary continuously after a fixed Cauchy shrink.
At the inner edge use the actual activation formulas
\eqref{eq:connection-flat-activation}. They give
$T_0=e_aB_0$ with the same flat scalar $e_a$ and a continuously varying
nonzero coefficient $B_0$. The limiting direction and its cone gap,
rather than the infimum of the vanishing raw stress, have a uniform
margin. The remaining fixed compact joining intervals have ordinary
strict relaxed-cone margins. All are preserved by reducing the input
radius.

The terminal discrepancy is measured against the \emph{already
pressure-edited} exterior. The five-moment inverse of
Lemma~\ref{lem:prepared-rank} removes its full accumulated defect,
including that due to $g$, exactly. The equalities in
\cite[Lemma~4.4]{OpenAI2026} then recover the actual exterior radial
flow, pressure and residual sources. Neither the prescribed core traces
nor the protected endpoint collars are changed by this solve.

\paragraph{Mean-preserving stress realization.}
There is also an equality in the source radial modulation, not just a
cone inequality. Let $S_0$ be its old input shear and let
$S_L(\theta)$ be its selected periodic loop with mean $S_0$. For the
new actual input $S$, take
\begin{equation}
 \widetilde S_L(\theta)=S_L(\theta)+(S-S_0),\qquad
 \langle\widetilde S_L\rangle=S.
 \label{eq:analytic-mean-preserving-loop}
\end{equation}
The old loop and actual residual have a strict cone margin on a compact
set. A small change of both preserves that margin. Where the old loop
equals its input near a support boundary, the translated loop equals
the new input exactly. Reconstruct its zero-mean primitives with the
new actual $E$. A common sufficiently large \emph{finite} modulation
frequency gives the source approximation bounds on the input ball;
restore the remaining small moment errors with the same five-moment
inverse. Thus the mean-shear equality, pressure and matching equations
are all retained, as required by \cite[Appendix~C]{OpenAI2026}.

There are only finitely many ball-invariance, inverse, positivity,
normalized-edge, joining and moment tolerances before the final
completion. Their minimum, divided by the finite bounds for the input
maps, gives $\eps_F>0$. Each input then has the source completion,
separately. This proves \eqref{eq:functionfamily} and
\eqref{eq:trace-surjectivity}; it does not assert continuity of a chosen
final non-axisymmetric field. Conditions $g(0)=g'(0)=0$ and fixed $F_*$
preserve the central first gradient, not the whole axis gradient.
The exact identities in \eqref{eq:app-slice-observables} follow from
$\delta u_z=\kappa q^{-A}\eta^3$ and
$\delta p_{zz}=[\psi''(0)-4A\psi(0)]\tau^{-2}$ at the stationary
center. The rectangle \eqref{eq:rectangle} lies strictly inside this
function-space ball and retains a positive strain-supply coefficient.

\subsection{Reconstruction of the pressure and offset comparisons}
\label{subsec:app-partition-reconstruction}
For completeness, the two $m=4$ comparisons do not require passage
through a degenerate exit. With $G=4\eta+j_0$, the constant part of the
axis angular source is $3-h+j_0\eta>2.8$ for small $j_0$.
The original even prepared pressure has $Z_*(\eta_H)>0$.
For the pressures \eqref{eq:app-beta}, instead
\[
 Z_{*,\beta}(\eta_H)=-(1+\beta)\calM_H<0.
\]
For any fixed compact $\beta$ interval these values are separated from
zero. The source core and continuation estimates use this separation
through $|Z_*|$ and $p_2^2/p_1$, not through the sign of $Z_*$. The
analytic inverse in \eqref{eq:core-fixed-point} likewise has no such
sign restriction. Reconstruct the reference member and the modified
compact family separately, choosing a common regularizer scale, then
common sufficiently large $\Lambda,C$. Since $H_*$ is identical, this
keeps the same entire prescribed $F_*$ and $G$. The tilts are
$O(j_0(K+\beta_{\max}))$ and hence fit the actual pressure edit for
small $j_0$. The noncritical angular source has a fixed positive lower
bound, so the source continuation and exact five-moment match apply,
with the edited exterior as target. No continuous interpolation through
$Z_*(\eta_H)=0$ has been asserted. This proves
Theorem~\ref{thm:families}(\ref{item:family-partition}).

For the signed-offset comparison, first take the $m=4$ tilted member of
Section~\ref{subsec:app-critical-connection}. At $j=0$, its transport
root has $Z_*(0)=-\gamma\ne0$. Changing $G$ to $4\eta+j$ is a small
analytic input to the same reconstruction just proved, without imposing
the central constraints $g(0)=g'(0)=0$. All finite inverse, exit,
normalized-edge and joining margins remain positive for $|j|\le j_*$
with some $j_*>0$. The moment target uses the same actual prepared
pressure, so it is restored exactly. The root and force computations
\eqref{eq:app-signed-root}--\eqref{eq:app-signed-coefficients} then prove
Theorem~\ref{thm:families}(\ref{item:family-offset}).

For all members, positivity of the swirl in the unchanged inner
stress-free region also gives velocity blowup when the center particle
is stationary: at $\eta=0$, any fixed $0<X<X_a$ has
$r=\sqrt{2X\tau}$ and leading speed $\tau^{-A}E(X,0)$.
Positive-order corrections are of higher relative order there. Thus
$\|u(t)\|_\infty\to\infty$ along these shrinking off-axis points;
nonzero center velocity is not a completion requirement.

\subsection{Viscosity and source dependence}\label{subsec:app-axis-viscosity}
All preceding calculations use $\nu=1$. For $\nu>0$, the viscosity-changing spatial and amplitude rescaling
\begin{equation}
 u_\nu(x,t)=\sqrt\nu\,u(x/\sqrt\nu,t),\qquad
 p_\nu(x,t)=\nu p(x/\sqrt\nu,t),\qquad
 f_\nu(x,t)=\sqrt\nu\,f(x/\sqrt\nu,t)
 \label{eq:app-viscosity-scaling}
\end{equation}
preserves the equation and the qualitative conclusions. The parameter bounds and support sizes change under this scaling, as expected. The complete-family statements are conditional on the external completion theorem and its pressure/moment interface; the axis identities themselves are direct consequences of the displayed equations.

\section{Completion interface and exact matching}\label{app:completion}

This appendix records the interface for the reference completion scheme $\rho_0$ in Section~\ref{subsec:compatibility}. Its pressure formula, angular moments and stress cone belong to the anisotropic representation of \eqref{eq:coordinates}; they are not imposed on every scheme in the co-design framework. The analytic and oscillatory completion is the external input in Assumption~\ref{ass:completion}. Its applicability to the deformed profiles is checked by the moment estimates below, the reconstruction in Appendix~\ref{app:axis}, and the angular-hierarchy extension in Appendix~\ref{app:angular}.

\subsection{The actual profile variables}\label{subsec:app-completion-variables}
With the coordinates in \eqref{eq:coordinates}, define
\begin{equation}
 \begin{aligned}
 M&=\int_0^XU\dd x,& I&=\int_0^X\sqrt{2x}E\dd x,\\
 J&=\int_0^XU\sqrt{2x}E\dd x,& S&=\int_0^X(U^2-E^2/2)\dd x .
 \end{aligned}
 \label{eq:app-moment-def}
\end{equation}
and the canonical pressure
\begin{equation}
 \Pi(X,\eta)=-\int_X^\infty\frac{E(x,\eta)^2}{2x}\dd x .
 \label{eq:app-pressure-def}
\end{equation}
The normalized quantities $\mu=M/X$, $\iota=I/(X\sqrt{2X})$, $\jmath=J/(X\sqrt{2X})$ and $\sigma=S/X$ satisfy the exact forward system
\begin{equation}
 \begin{aligned}
 E_y&=(1-a)E/2,& U_y&=bE/2,\\
 \mu_y&=U-\mu,&\iota_y&=E-3\iota/2,\\
 \jmath_y&=UE-3\jmath/2,&\sigma_y&=U^2-E^2/2-\sigma,
 \qquad\Pi_y=E^2/2.
 \end{aligned}
 \label{eq:app-forward}
\end{equation}
Here $y=\log X$, $a=1-2E_y/E$ and $b=2U_y/E$ are design variables on a finite interval. The radial velocity is not an independent variable:
\begin{equation}
 ru_r^{(0)}=\frac X L(2\eta U-2D\eta\mu-d\mu_\eta).
 \label{eq:app-radial}
\end{equation}

The actual source coefficients are
\begin{align}
 W&=1-2D\eta\mu-d\mu_\eta,\nonumber\\
 Q&=-W+\frac{(1-h)\iota-D\eta\iota_\eta-d\jmath_\eta+2(h-D)\eta\jmath}{E},\label{eq:app-Q}\\
 N&=-WU+D(\mu-\eta\mu_\eta)+4h\eta\sigma-d\sigma_\eta
              +4A\eta\Pi-d\Pi_\eta .\label{eq:app-N}
\end{align}
These formulas show why pressure and parameter derivatives cannot be changed after a velocity profile has been selected. The stress coordinates are
\begin{equation}
 \begin{aligned}
 p_1&=\frac{XQ}{L},& p_2&=\frac{XN}{LE},&
 \vartheta&=-\frac ba,\\
 v&=a(1+\vartheta^2),&
 P_c&=p_1+\vartheta p_2,&
 J_c&=p_2-\vartheta p_1 .
 \end{aligned}
 \label{eq:app-stress-coordinates}
\end{equation}

\subsection{The admissible stress interface}\label{subsec:app-completion-stress}
The reference realization uses a stress-free inner core and an active annulus. On the annulus the sufficient cone conditions are
\begin{equation}
 a>0,\qquad v>2,\qquad P_c>v,\qquad
 (v-2)J_c^2<2(P_c-v)^2 .
 \label{eq:app-cone}
\end{equation}
The actual leading tangential stress is
\begin{equation}
 T_0=\frac E{\sqrt{2X}}\bigl((p_1,p_2)-(a,-b)\bigr).
 \label{eq:app-stress}
\end{equation}
At the inner and outer edges, the stress is multiplied by a prescribed flat factor. Consequently the normalized direction, rather than the vanishing raw stress, is the quantity that must have a positive limiting margin. In the inner core, $T_0=0$ and no division by $a$ is required. This is why a subcritical member may have a short interval of negative core shear while satisfying the annular condition.

We record the following implication of the stated source-dependent completion input; it is not a new scheme-independent completion theorem.

\begin{proposition}[Completion interface]\label{prop:interface}
Under Assumption~\ref{ass:completion}, let a leading profile have a regular analytic axis core, positive $E$ for $X>0$, a flat stress-free inner edge, an active annulus satisfying \eqref{eq:app-cone} with a smooth limiting stress direction, and an exact prepared exterior. Assume also that its ordered correction intervals have the required moment right inverses and that all parameter, higher-order and localization bounds required by that input hold for the actual profile. Then the reference scheme gives
\[
 \mathbf{Comp}_{\rho_0}(\mathbf d)\ne\varnothing.
\]
In particular there is a smooth incompressible Navier--Stokes field satisfying \eqref{eq:NS}--\eqref{eq:target}, with the late axis traces and first gradient preserved as specified in Assumption~\ref{ass:completion}.
\end{proposition}

Proposition~\ref{prop:interface} records the implication used from \cite{OpenAI2026}; it does not prove that an arbitrary profile deformation meets its hypotheses. For Appendix~\ref{app:axis}, the reconstruction has to control analytic axis changes on a common complex neighborhood, restore the actual pressure and all accumulated moment defects, and preserve normalized flat-edge margins. The finite choices follow the source-dependent order: outer data and strict margins first, then the analytic core scale, then the physical amplitude and localized corrections. The local moment inverse is combined with the axis-core and global connection arguments in Sections~\ref{subsec:app-core-contraction}--\ref{subsec:app-local-lifting}. For the nonzero angular-budget family, Proposition~\ref{prop:angular-completion} states and proves the change to the source background construction before its downstream completion is invoked.

\subsection{The five-moment correction}\label{subsec:app-completion-moments}
Use $R=\sqrt{2X}$ and the pressure-moment convention
\begin{equation}
 C_p(\eta):=\int_0^\infty \frac{E(R,\eta)^2}{R}\,\dd R
          =-\Pi_{\rm ax}(\eta).
 \label{eq:app-pressure-moment}
\end{equation}
Here $E(R,\eta)$ abbreviates the same profile expressed in $R$ coordinates. Consequently a prescribed axis pressure increment $\psi$ requires $\Delta C_p=-\psi$, with $\Delta$ denoting a finite moment increment. On a prepared interval, let $(\delta E,\delta U)$ be compactly supported changes. The leading moment differential is
\begin{equation}
 \mathcal L_\eta(\delta E,\delta U)=
 \begin{pmatrix}
 \int R^2\delta E\dd R\\
 \int R\delta U\dd R\\
 \int (2E/R)\delta E\dd R\\
 \int R^2(U\delta E+E\delta U)\dd R\\
 \int R(2U\delta U-E\delta E)\dd R
 \end{pmatrix},
 \qquad R=\sqrt{2X}.
 \label{eq:app-five-map}
\end{equation}
The rows represent, in order, the angular momentum, axial mass, $C_p$, mixed angular momentum and energy-type moment. A finite set of compactly supported vector bumps $B_j$ gives a matrix $\mathsf B(\eta)$ with columns $\mathcal L_\eta B_j$. If
\begin{equation}
 \operatorname{rank}\mathsf B(\eta)=5,\qquad
 \lambda_{\min}(\mathsf B\mathsf B^T)\ge c_*>0
 \quad(-1\le\eta\le1),
 \label{eq:app-rank}
\end{equation}
then
\begin{equation}
 \mathsf R(\eta)=\mathsf B(\eta)^T[\mathsf B(\eta)\mathsf B(\eta)^T]^{-1}
 \label{eq:app-right-inverse}
\end{equation}
is a smooth right inverse. The nonlinear moment increment is $\mathsf Bc+\mathcal Q(c,c)$. The target is the full accumulated defect relative to the required terminal moments, not just the defect of the last design step. Thus a pressure-only restoration has target $(0,0,-\psi,0,0)^T$, while a preceding core or transport change can require other nonzero components.

\begin{lemma}[Rank on the actual prepared patches]\label{lem:prepared-rank}
Let $x=X/\varrho$ on a fixed prepared interval, and suppose that
\[
 U=u_c(\eta),\qquad E=e_*f(\eta)x^\alpha,
 \qquad e_*>0,\quad \inf_{[-1,1]}f>0.
\]
If $\alpha\notin\{-\tfrac12,\tfrac12,\tfrac32\}$, two axial and three
azimuthal bumps can be chosen so that the five-moment differential has
rank five. After the natural radial and amplitude normalizations, its
inverse and every fixed parameter derivative are uniformly bounded on
compact families of these data. This applies to the joining patch
$\alpha=1/10$ and to the protected patches
$\alpha=-1/2-\lambda$ with fixed $\lambda>0$.
\end{lemma}
\begin{proof}
Use the moment order $(M,I,J,S,C_p)$ temporarily. Divide out the nonzero
powers of $\varrho,e_*,f$ and subtract $u_c$ times the $I$ row from the
$J$ row and $2u_c$ times the $M$ row from the $S$ row. The axial and
azimuthal blocks then have, respectively, the weights
\begin{equation}
 (1,x^{\alpha+1/2}),\qquad
 (x^{1/2},x^\alpha,x^{\alpha-1}).
 \label{eq:prepared-block-weights}
\end{equation}
Choose two, respectively three, ordered disjoint compact subintervals of
the prepared interval. On each choose a nonzero nonnegative smooth bump;
for example the pullback of $\exp[-1/(1-s^2)]\mathbf1_{|s|<1}$.
For distinct exponents $a_1,\ldots,a_k$, the determinant of their response
matrix is the integral, over the ordered product of these supports, of
$\det[x_j^{a_i}]$ times the product of the bumps. This determinant has a
constant nonzero sign: a nonzero linear combination of $k$ distinct powers
has at most $k-1$ positive zeros. The last assertion follows by dividing
by the smallest power, differentiating, and inducting with Rolle's
theorem. Thus both blocks are invertible. This is the power-moment
argument of \cite[Lemma~4.7 and Corollary~A.3]{OpenAI2026}, applied here to
the actual background rather than to unspecified test functions.

For $\alpha=1/10$ the powers are $(0,3/5)$ and
$(1/2,1/10,-9/10)$. For $\alpha=-1/2-\lambda$ they are
$(0,-\lambda)$ and $(1/2,-1/2-\lambda,-3/2-\lambda)$.
The determinants stay away from zero on the fixed compact parameter sets.
The row operations and their inverses have bounded parameter derivatives;
the same is therefore true of the normalized inverse. The physical
unscaled inverse is also finite once $\varrho$ and the amplitudes have
been fixed. No bound uniform as $\lambda\downarrow0$ is asserted.
\end{proof}

The joining construction uses the first patch in
Lemma~\ref{lem:prepared-rank}. Pressure-only editing uses three azimuthal
bumps in a separate $U=0$ power patch of the preparation; the unused gap
between the reserved leading correction intervals in
\cite[Appendix~A.9]{OpenAI2026} can be subdivided for this purpose. The
later $I_{\rm pos}$ and $I_{\rm mean}$ patches are not edited. On a
$U=0$ patch, azimuthal changes leave $M,J$ identically unchanged; the
three remaining rows are exactly the second block of
\eqref{eq:prepared-block-weights}. Hence a sufficiently small pressure
trace change is realized by actual swirl changes satisfying
\eqref{eq:app-pressure-matching}, with the quadratic terms retained.
All supports are fixed before taking parameter neighborhoods. The
five-moment solve after the core connection restores its accumulated
five-component defect relative to this already pressure-edited exterior.

\paragraph{Parameter-uniform moment restoration.}
The same finite correction can accommodate function-valued targets. Suppose that $\mathsf B$, $\mathsf R$ and the coefficients of $\mathcal Q$ extend boundedly and holomorphically to a fixed complex neighborhood $\Om$ of $[-1,1]$, with $\mathsf B\mathsf R=\mathrm{Id}$. A real-axis rank bound alone is not this analytic hypothesis. If $\mathsf B$ is holomorphic near the interval, the nonvanishing real-axis determinant of $\mathsf B\mathsf B^T$ permits a smaller common neighborhood on which \eqref{eq:app-right-inverse} is holomorphic; $T$ here is the algebraic transpose, not conjugate transpose.

Let $M_R$ bound the right inverse and let $C_Q$ be a bilinear bound for $\mathcal Q$ in the corresponding vector-valued $H^\infty(\Om)$ norms. For a sufficiently small holomorphic target $\delta m$, solve
\begin{equation}
 c=\mathsf R\,\delta m-\mathsf R\,\mathcal Q(c,c),
 \qquad \|c\|_\Om\le 2M_R\|\delta m\|_\Om.
 \label{eq:app-analytic-moment-lift}
\end{equation}
Indeed, on the ball of radius $r=2M_R\|\delta m\|_\Om$, the fixed-point map has Lipschitz constant at most $2M_RC_Qr$ and maps the ball into itself when $4M_RC_Qr<1$. Applying $\mathsf B$ to the fixed-point equation gives the exact nonlinear moment equation. Reality on the real interval is preserved by the iteration. For $\Om'\Subset\Om$, Cauchy estimates give, at each fixed order $k$,
\begin{equation}
 \|\partial_\eta^k c\|_{\Om'}
 \le k!\,\operatorname{dist}(\Om',\partial\Om)^{-k}
             \|c\|_\Om.
 \label{eq:app-analytic-parameter-bound}
\end{equation}
For the prepared patches in Lemma~\ref{lem:prepared-rank}, the
coefficient functions are holomorphic. Compactness and nonvanishing of
the determinant on the real interval give one smaller complex
neighborhood on which the square block inverses are holomorphic and
bounded. Thus the analytic hypotheses just used hold for those actual
patches. A finite number of such neighborhoods can be intersected. This
proves the parameter-uniform moment-restoration step. The separate core
and connection arguments are supplied in
Sections~\ref{subsec:app-core-contraction}--\ref{subsec:app-local-lifting};
no continuous selection of the final infinite correction is required.

The same argument allows an actual leading axial velocity in the correction interval. It does not require the convenient special choice $U=0$ or a pure-power rotation profile, provided the full rank and strict cone bounds are verified for the actual background. Thus the five-moment construction is a compatibility tool, not a physical restriction on all designs.

\subsection{Terminal heat profile and canonical pressure}\label{subsec:app-completion-exterior}
The exterior used by the source has a positive heat factor $\mathcal H$ and a power-law asymptotic profile
\begin{equation}
 E_{\rm heat}(X,\eta)=c_\infty X^{-A}\mathcal H(2d/X),\qquad
 \mathcal H(Z)=\frac1{\Gamma(1+h)}\int_0^\infty e^{-s}s^h(1+Zs)^{-h}\dd s .
 \label{eq:app-heat}
\end{equation}
The terminal requirements are not independent assignments: $M(\infty)=S(\infty)=0$, the angular transport equation, and \eqref{eq:app-pressure-def} determine the remaining pressure and moment targets together. A small pressure perturbation in the axis trace is therefore matched by the finite-dimensional change of $E$ described above. After the correction, the exterior heat factor and its flat endpoint are unchanged.

\subsection{Why the interface leaves a broad design space}\label{subsec:app-completion-scope}
The interface constrains actual equations, pressure, moments, edge regularity and the reference realization cone; it does not prescribe a unique profile. Distinct members can lie \emph{within the same} observation fiber. The pressure--viscosity comparison fixes the prescribed axis velocity and first gradient; the analytic family fixes only the central kinematics; the stretching family fixes the central vorticity trace while varying the strain. Their observation maps therefore need to be specified separately.

The resulting non-rigidity is within the reference scheme $\rho_0$. The broader formulation in Section~\ref{subsec:compatibility} allows other schemes without claiming that this cone, moment map or pressure identity completes them. Changing geometry or the realization mechanism requires verification of the corresponding interface rather than automatic reuse of Proposition~\ref{prop:interface}.

\section{Joint angular transport and its free kernel}\label{app:angular}

In the anisotropic representation of Section~\ref{subsec:explicit}, the angular compatibility condition is one coupled equation for a pair of moment functions, not two independent zero-moment requirements. This appendix identifies its complete smooth kernel and right inverse. These are algebraic compatibility results; their use in the completed families of Section~\ref{sec:families} additionally requires realization by actual velocity profiles through Appendix~\ref{app:completion}.

\subsection{All-order angular operator}\label{subsec:app-angular-operator}
Let $I_n$ denote the angular moment at order $n\ge0$ and let $J_n$ denote the complete mixed angular flux at that order; the principal moment uses the renormalization in \eqref{eq:angular-moments}. Set
\begin{equation}
 b_n=1-h+2nh,\qquad c_n=\frac12-2h+2nh,\qquad
 \Ax_n(I,J)=D\eta I'-b_nI+dJ'+2c_n\eta J .
 \label{eq:app-An}
\end{equation}
The full flux $J_n$ includes the contributions of every product $U_iE_j$ with $i+j=n$. The principal equation is $\Ax_0(I_0,J_0)=0$. For $n\ge1$, the source hierarchy used here has the form
\begin{equation}
 \Ax_n(I_n,J_n)=\nu LZ_{c_n}Z_{b_{n-1}}I_{n-1},
 \qquad Z_bf=\frac{2b\eta f+df'}L .
 \label{eq:app-hierarchy}
\end{equation}
Since $LZ_{c_n}=d\partial_\eta+2c_n\eta$, the right side is removed by the shift
\begin{equation}
 \Ax_n\left(I_n,J_n-\nu Z_{b_{n-1}}I_{n-1}\right)=0 .
 \label{eq:app-shift}
\end{equation}
This identity is the all-order form of the first-order angular-flux compensation.

\subsection{An explicit right inverse and the complete kernel}\label{subsec:app-angular-kernel}
Write $k_n=b_n+D=1+c_n$. For $P_n=D\eta I+dJ$, the coefficient relations give the identity
\[
 \Ax_n(I,J)=P_n'-k_n(I-2\eta J).
\]
This identity makes the cancellation in the following right inverse explicit:
\begin{equation}
 \mathcal R_ng=\left(-\frac{dg}{k_nL},\frac{D\eta g}{k_nL}\right),
 \qquad \Ax_n\mathcal R_ng=g .
 \label{eq:app-transport-right-inverse}
\end{equation}
For every smooth $\phi$, define
\begin{equation}
 \Kern_n\phi=\left(\frac{d\phi'+2k_n\eta\phi}{L},
                         \frac{k_n\phi-D\eta\phi'}{L}\right).
 \label{eq:app-transport-kernel}
\end{equation}
The two relations $D\eta I+dJ=k_n\phi$ and $I-2\eta J=\phi'$ have determinant $L=d+2D\eta^2$. Solving them gives \eqref{eq:app-transport-kernel}, and the preceding identity gives $\Ax_n\Kern_n\phi=0$. Conversely, if $\Ax_n(I,J)=0$, the function
\begin{equation}
 \phi=\frac{D\eta I+dJ}{k_n}
 \label{eq:app-kernel-potential}
\end{equation}
recovers $(I,J)=\Kern_n\phi$. Hence all smooth solutions of an inhomogeneous equation are
\begin{equation}
 (I,J)=\mathcal R_ng+\Kern_n\phi,\qquad \phi\in C^\infty([-1,1]) .
 \label{eq:app-all-solutions}
\end{equation}
No endpoint or center resonance is introduced by this joint parameterization: only $L$ and $k_n$ are divided by, and both are bounded away from zero for $0<h<1/100$.

\subsection{A nonzero principal angular budget}\label{subsec:app-angular-budget}
For the principal order, let $k_0=3/2-2h$. One useful kernel element is
\begin{equation}
 I_0=c\left((1-h)^{-1}-4\eta^2\right),\qquad J_0=c\eta .
 \label{eq:app-principal-kernel}
\end{equation}
The first positive order can then carry the required viscous transport by choosing
\begin{equation}
 I_n=0\ (n\ge1),\qquad
 J_1=\nu Z_{1-h}I_0,\qquad J_n=0\ (n\ge2).
 \label{eq:app-first-flux}
\end{equation}
The explicit first flux is
\begin{equation}
 J_1(\eta)=\nu c\,\frac{\eta(-6+8h\eta^2)}{1-2h\eta^2} .
 \label{eq:app-J1}
\end{equation}
These formulas solve the angular hierarchy at the level of moment targets. For the small nonzero-$c$ completed family asserted in Section~\ref{sec:families}, $J_1$ must be the actual mixed flux of the corrected velocity, including all products at that order. The target is realized through the moment corrections in Proposition~\ref{prop:angular-completion}, with the axis data, canonical pressure, cone and remaining exterior conditions retained. Assigning $J_1$ as a free scalar function is not by itself a full-flow construction. Zero angular moments are one compatibility choice, not a universal physical necessity.

\subsection{Realization of the angular targets}\label{subsec:app-angular-realization}
We first derive the cancellation needed by the background construction.
At relative order $q^{2nh}$, define the physical angular inventory and
complete axial flux by
\begin{equation}
 \mathscr L_n=q^{b_n}I_n,\qquad \mathscr J_n=q^{c_n}J_n.
 \label{eq:angular-physical-moments}
\end{equation}
The factor in $\mathscr L_n$ is obtained by integrating
$r^2u_{\theta,n}\,\dd r$; the flux includes all products of total order
$n$, not just $U_0E_n$. The principal inventory is renormalized by the
$z$- and $t$-independent pure-power term in
\eqref{eq:angular-moments}. Conservative angular momentum gives
\begin{equation}
 \partial_t\mathscr L_0+\partial_z\mathscr J_0=0,\qquad
 \partial_t\mathscr L_n+\partial_z\mathscr J_n
       =\nu\partial_{zz}\mathscr L_{n-1}\quad(n\ge1).
 \label{eq:angular-conservative-hierarchy}
\end{equation}
For any exponent $b$, the coordinate identities imply
$\partial_z(q^bf)=q^{b-D}Z_bf$ and
$\partial_t(q^bf)=q^{b-1}(D\eta f'-bf)/L$.
Since $b_{n-1}-D=c_n$, substitution gives exactly
\eqref{eq:app-hierarchy}. In particular the two successive $z$
derivatives use $Z_{c_n}Z_{b_{n-1}}$, not two copies of the same
operator. Equation~\eqref{eq:app-first-flux} solves this hierarchy at
all orders: its only nonzero positive-order flux is the viscous return
of $I_0$.

\begin{proposition}[Completion with a joint angular budget]
\label{prop:angular-completion}
Assume the validity of the source results in
Assumption~\ref{ass:completion}. Suppose an actual leading profile has
the regularity, normalized cone and flat-edge properties, exact terminal
collar, and protected correction patches of the source construction.
Replace only its two zero angular moments by an actual smooth pair
$(I_0,J_0)$ satisfying $\Ax_0(I_0,J_0)=0$, while retaining
$M(\infty)=S(\infty)=0$ and the canonical pressure. Then the target
sequence \eqref{eq:app-first-flux} can be realized by actual higher-order
velocity coefficients. The source background estimates and downstream
completion remain valid. Moreover the pair in
\eqref{eq:app-principal-kernel} is produced by an actual leading-profile
deformation for every sufficiently small $|c|$. Consequently these
nonzero-budget profiles have complete singular realizations with their
original late core traces.
\end{proposition}
\begin{proof}
It is enough to work with viscosity one and apply
\eqref{eq:app-viscosity-scaling} at the end. We give separately the
leading deformation, the positive-order moment solve, and the residual
cancellation on which completion depends.

\emph{Leading deformation.}
Use a compact $U=0$ power patch strictly inside the active annulus and
disjoint from the protected $I_{\rm pos},I_{\rm mean}$ intervals and
both endpoint collars. The finite leading preparation leaves such an
interval available, as in Appendix~\ref{app:completion}. Solve the
nonlinear five-moment increment with targets, in the order
$(M,I,J,S,C_p)$,
\[
 (0,I_0,J_0,0,0).
\]
Lemma~\ref{lem:prepared-rank} and
\eqref{eq:app-analytic-moment-lift} give genuine $(\delta E,\delta U)$
for all sufficiently small $|c|$. The core pressure is unchanged because
the total $C_p$ increment is zero. Beyond the editing interval the local
velocities and the cumulative $M,S$ are unchanged, while $I,J$ are
shifted by their prescribed targets. Using the actual formula
\eqref{eq:app-Q}, the resulting difference is exactly
\begin{equation}
 \delta Q=-\frac{\Ax_0(I_0,J_0)}{X\sqrt{2X}\,E}=0,
 \qquad \delta N=0.
 \label{eq:angular-exterior-cancellation}
\end{equation}
Thus the original terminal collar and stress are recovered, not merely
their limiting values. On the compact editing region, the strict cone
is preserved by smallness in the finite parameter norms used by the
source. The leading core, both flat edges and both protected patches
are untouched. This supplies the actual leading profile in the statement.

\emph{Higher-order velocities.}
At each positive order, use the analytic inner solve of
\cite[Lemma~5.1]{OpenAI2026} with its zero axis values and the already
finalized lower-order coefficients. Its sparse Volterra system supplies
one fixed radial interval; the parameter neighborhood may shrink with
order, exactly as in that lemma. In the global extension of
\cite[Lemma~5.2]{OpenAI2026}, replace the five zero moment targets,
in the source's order, by
\begin{equation}
 m_n=(0,I_n,0,J_n,0).
 \label{eq:angular-high-order-target}
\end{equation}
The current-order equations are affine linear: every nonlinear product
not containing the current unknown is already part of the defect.
On $I_{\rm pos}$ one has
$U_0=0$, $E_0=e_*f(\eta)R^{-1-2\lambda}$.
With the two axial and three azimuthal bumps used there, the exact
coefficient equations are
\begin{equation}
 \begin{aligned}
 B_U\alpha_n&=-d_{U,n}
              +\begin{pmatrix}0\\J_n/(e_*f)\end{pmatrix},\\
 B_E\beta_n&=-d_{E,n}
              +\begin{pmatrix}I_n\\0\\0\end{pmatrix}.
 \end{aligned}
 \label{eq:angular-realized-bumps}
\end{equation}
Here $d_{U,n},d_{E,n}$ are the complete source defects, including all
lower-order products and the radial pressure source. The matrices are
the fixed invertible power-moment blocks of
Lemma~\ref{lem:prepared-rank}, with the source's row normalizations.
These linear solves require no smallness of a target at a given order.
Their constants can depend on $n$ and on each fixed derivative order.
In particular $J_1$ is now a moment of the actual corrected velocities;
its contribution to later pressure and nonlinear defects is retained.

The zero first component in \eqref{eq:angular-high-order-target}
retains compact streamfunctions and radial velocities. The zero third
component retains the normalized pressure support. Thus all
positive-order velocity coefficients and their potentials still vanish
beyond the common $X_+$ and on $I_{\rm mean}$. Their axis traces and
first Cartesian gradients are the original zero traces. The zero fifth
component retains the axial momentum cancellation.

\emph{Stress primitives and completion.}
The only use of separate angular zero targets in
\cite[Lemma~5.2, Step~4]{OpenAI2026} is to cancel the total weighted
angular residual integral. Replace that step by
\eqref{eq:angular-conservative-hierarchy}. At every order it gives
\[
 \int_0^\infty R^2 r_{\theta,n}\,\dd R=0.
\]
The axial residual integral still vanishes by the unchanged first and
fifth moment targets. Consequently both tangential stress primitives
have their forward and backward representations. For $n\ge2$, every
exterior same-order product contains a positive-order field, and the
axial viscosity acts on order $n-1>0$; hence the stress vanishes beyond
the same $X_+$. At $n=1$, the sole exterior source is the original
$-\partial_{zz}u_{\theta,0}$. The unchanged terminal collar therefore
gives exactly the outer flat-weight estimate of
\cite[Equation~(5.22)]{OpenAI2026}. The subtracted principal integral
can be differentiated by the same heat-exterior estimate as in that
proof; its nonzero value is compensated by the actual $J_1$.

These facts are precisely the support and fixed-order bounds used in
\cite[Proposition~5.3, Lemma~5.4 and Proposition~5.5]{OpenAI2026}.
Summation is performed on the potentials before taking curls, preserving
incompressibility. Its $q$ cutoffs need not preserve each formal flux
identity separately: comparison with any sufficiently long finite
truncation gives a residual of arbitrarily high order, as in Lemma~5.4.
This argument allows the fixed-order constants to grow with $n$.
It therefore retains the smooth, flat remainder rather than imposing an
unjustified identity on the cutoff sum.

The resulting background has the same annulus, normalized stress
bounds, exact heat collar, compact potentials and exact velocity on
$I_{\rm mean}$. These are the inputs used by
\cite[Sections~6--10]{OpenAI2026}. The zero moments required there for
\emph{added mean corrections} are unchanged; they are not conditions
that the modified leading angular inventory must vanish. Localization
uses the zero axial moments, which have been retained at every order.
The downstream stress realization and completion therefore apply under
Assumption~\ref{ass:completion}. The core and its blowup are unchanged,
and the force has the same support, regularity and flatness properties.
\end{proof}

\subsection{Use in the global design problem}\label{subsec:app-angular-design}
For a chosen terminal angular potential $\phi_0$, the principal pair is
\begin{equation}
 (I_0,J_0)=\left(\frac{d\phi_0'+2k_0\eta\phi_0}{L},
                  \frac{k_0\phi_0-D\eta\phi_0'}{L}\right).
 \label{eq:app-terminal-kernel}
\end{equation}
The global angular condition is then a single compatibility equation for the actual pair $(I_0,J_0)$, rather than two independent zero targets. The remaining moment conditions, the canonical pressure and the stress cone are still evaluated from the same $E,U$ and their parameter derivatives. This distinction matters: the kernel enlarges the design space, but it does not permit pressure or stress to be assigned independently of the velocity history.

For smooth inputs and fixed $h,n,m$, the explicit formulas give
\begin{equation}
 \|\mathcal R_n g\|_{C^m\times C^m}
       \le C_{h,n,m}\|g\|_{C^m},\qquad
 \|\Kern_n\phi\|_{C^m\times C^m}
       \le C_{h,n,m}\|\phi\|_{C^{m+1}}.
 \label{eq:app-transport-bounds}
\end{equation}
The kernel parameterization differentiates its potential once; a same-order $C^m$ bound for that map is not asserted. No uniform estimates as $h\downarrow0$ or $n\to\infty$ are needed for these identities, nor do they provide all-order completion bounds. Every order used by the completion has to satisfy the corresponding source-dependent estimates.

\subsection{Relation to the two realization strategies}\label{subsec:app-angular-strategies}
A profile-led strategy proposes $(E,U,\Pi)$ and seeks compatible transport and stress realization. A dynamics-led strategy begins with a background and a mechanism producing the needed stress or transport. When both use the same anisotropic representation and moment definitions, the angular operator and the moment map above are common compatibility objects. In another geometry, their formulas must be derived and checked anew.

These are different entry points into the coupled compatibility--realization problem, not classes of solutions for which an inclusion or non-inclusion theorem is asserted. A directly chosen profile can require amplified corrections, and a dynamically generated stress must satisfy the actual pressure and transport balances. The smooth angular kernel is one explicit compatibility freedom within this representation, not an arbitrary family of completed singular cores.

\bibliographystyle{unsrt}
\bibliography{references}

@article{Leray1934,
  author = {Jean Leray},
  title = {Sur le mouvement d\textquotesingle un liquide visqueux emplissant l\textquotesingle espace},
  journal = {Acta Mathematica},
  volume = {63},
  pages = {193--248},
  year = {1934},
}

@article{CKN1982,
  author = {Luis Caffarelli and Robert Kohn and Louis Nirenberg},
  title = {Partial regularity of suitable weak solutions of the {Navier--Stokes} equations},
  journal = {Communications on Pure and Applied Mathematics},
  volume = {35},
  number = {6},
  pages = {771--831},
  year = {1982},
}

@article{ESS2003,
  author = {Luis Escauriaza and Gregory A. Seregin and Vladimir \v{S}ver\'{a}k},
  title = {{$L_{3,\infty}$}-solutions of {Navier--Stokes} equations and backward uniqueness},
  journal = {Russian Mathematical Surveys},
  volume = {58},
  number = {2},
  pages = {211--250},
  year = {2003},
}

@misc{Fefferman2000,
  author = {Charles L. Fefferman},
  title = {Existence and smoothness of the {Navier--Stokes} equation},
  howpublished = {Clay Mathematics Institute Millennium Prize Problem statement},
  year = {2000},
}

@article{DLS2009,
  author = {Camillo De Lellis and L\'{a}szl\'{o} Sz\'{e}kelyh\'{i}di, Jr.},
  title = {The {Euler} equations as a differential inclusion},
  journal = {Annals of Mathematics},
  volume = {170},
  number = {3},
  pages = {1417--1436},
  year = {2009},
}

@article{Daneri2017,
  author = {Sara Daneri and L\'{a}szl\'{o} Sz\'{e}kelyh\'{i}di, Jr.},
  title = {Non-uniqueness and {$h$}-principle for {H\"older}-continuous weak solutions of the {Euler} equations},
  journal = {Archive for Rational Mechanics and Analysis},
  volume = {224},
  pages = {471--514},
  year = {2017},
}

@article{BuckmasterVicol2019,
  author = {Tristan Buckmaster and Vlad Vicol},
  title = {Nonuniqueness of weak solutions to the {Navier--Stokes} equation},
  journal = {Annals of Mathematics},
  volume = {189},
  number = {1},
  pages = {101--144},
  year = {2019},
}

@article{Albritton2022,
  author = {Dallas Albritton and Elia Bru\'{e} and Maria Colombo},
  title = {Non-uniqueness of {Leray} solutions of the forced {Navier--Stokes} equations},
  journal = {Annals of Mathematics},
  volume = {196},
  number = {1},
  pages = {415--455},
  year = {2022},
}

@article{Tao2016,
  author = {Terence Tao},
  title = {Finite time blowup for an averaged three-dimensional {Navier--Stokes} equation},
  journal = {Journal of the American Mathematical Society},
  volume = {29},
  number = {3},
  pages = {601--674},
  year = {2016},
}

@article{Lifschitz1991,
  author = {Alexander Lifschitz and Eliezer Hameiri},
  title = {Local stability conditions in fluid dynamics},
  journal = {Physics of Fluids A: Fluid Dynamics},
  volume = {3},
  number = {11},
  pages = {2644--2651},
  year = {1991},
}

@article{Friedlander1991,
  author = {Susan Friedlander and Misha M. Vishik},
  title = {Instability criteria for the flow of an inviscid incompressible fluid},
  journal = {Physical Review Letters},
  volume = {66},
  number = {17},
  pages = {2204--2206},
  year = {1991},
}

@article{Craik1986,
  author = {A. D. D. Craik and W. O. Criminale},
  title = {Evolution of wavelike disturbances in shear flows: a class of exact solutions of the {Navier--Stokes} equations},
  journal = {Proceedings of the Royal Society of London. Series A},
  volume = {406},
  number = {1830},
  pages = {13--26},
  year = {1986},
}

@article{Singh2017,
  author = {Nishant K. Singh and S. Sridhar},
  title = {Plane shearing waves of arbitrary form: exact solutions of the {Navier--Stokes} equations},
  journal = {The European Physical Journal Plus},
  volume = {132},
  pages = {403},
  year = {2017},
}

@misc{Cordoba2023,
  author = {Diego C\'{o}rdoba and Luis Mart\'{i}nez-Zoroa},
  title = {Blow-up for the incompressible {3D Euler} equations with uniform {$C^{1,1/2-\epsilon}\cap L^2$} force},
  year = {2023},
  note = {Preprint},
}

@article{Cordoba2026,
  author = {Diego C\'{o}rdoba and Luis Mart\'{i}nez-Zoroa and Fan Zheng},
  title = {Finite time blow-up for the hypodissipative {Navier--Stokes} equations with a force in {$L^1_t C^{1,\epsilon}_x\cap L^\infty_tL^2_x$}},
  journal = {Archive for Rational Mechanics and Analysis},
  volume = {250},
  pages = {38},
  year = {2026},
}

@misc{OpenAI2026,
  author = {{OpenAI}},
  title = {Finite time blowup for {Navier--Stokes}},
  year = {2026},
  howpublished = {Public manuscript, released 8 September 2026},
  note = {166-page supplied version; used as a conditional completion input},
  url = {https://cdn.openai.com/pdf/32d9f210-8b73-45e0-91bc-82a30aef8a9a/navier-stokes.pdf},
}

@book{ConstantinFoias1988,
  author = {Peter Constantin and Ciprian Foias},
  title = {{Navier--Stokes} Equations},
  publisher = {University of Chicago Press},
  year = {1988},
}

@article{ConstantinFefferman1993,
  author = {Peter Constantin and Charles Fefferman},
  title = {Direction of vorticity and the problem of global regularity for the {Navier--Stokes} equations},
  journal = {Indiana University Mathematics Journal},
  volume = {42},
  number = {3},
  pages = {775--789},
  year = {1993},
}
\end{document}